\documentclass[11pt]{article}

\usepackage[colorlinks=true,linkcolor=blue,citecolor=blue]{hyperref}

\usepackage{amsmath, amssymb, amsthm}
\usepackage{mathtools}
\usepackage[capitalize,nameinlink]{cleveref}
\usepackage{fullpage}
\usepackage{graphics}
\usepackage{pifont}
\usepackage{tikz}
\usepackage{bbm}

\usetikzlibrary{arrows.meta}
\newcommand{\xmark}{\text{\ding{55}}}%
\usepackage{environ}
\usepackage{framed}
\usepackage{url}
\usepackage{algorithm}
\usepackage[noend]{algpseudocode}
\usepackage[labelfont=bf]{caption}
\usepackage{cite}
\usepackage{framed}
\usepackage[framemethod=tikz]{mdframed}
\usepackage{appendix}
\usepackage{graphicx}
\usepackage[textsize=tiny]{todonotes}
\usepackage{tcolorbox}

\newcommand\remove[1]{}

\newtheorem{theorem}{Theorem}
\newtheorem{lemma}{Lemma}[section]
\newtheorem*{lemma*}{Lemma}
\newtheorem{corollary}[lemma]{Corollary}
\newtheorem*{corollary*}{Corollary}
\newtheorem{claim}[lemma]{Claim}

\theoremstyle{definition}

\newtheorem*{theorem*}{Theorem}

\newtheorem*{rem*}{Remark}

\newcommand\E{\mathbb{E}}

\newcommand{\eps}{\varepsilon}
\renewcommand{\O}{\widetilde{O}}%{\tilde{O}}
\newcommand{\Seq}{\textsc{Seq}}

\newcommand{\pe}{\preceq}

\renewcommand{\l}{\langle}
\renewcommand{\r}{\rangle}
\newcommand{\bs}{\backslash}

\newcommand{\hm}{\hat{m}}
\newcommand{\assign}{\leftarrow}
\newcommand{\br}{\bar{r}}
\newcommand{\head}{\mathrm{head}}
\newcommand{\tail}{\mathrm{tail}}
\newcommand{\ParallelSC}{\textsc{ParallelSC}}
\newcommand{\ParallelDiam}{\textsc{ParallelDiam}}
\newcommand{\DistrReach}{\textsc{DistrReach}}
\newcommand{\kd}{\kappa D}
\newcommand{\otilde}{\O}
\newcommand{\Anc}{\mathrm{Anc}}
\newcommand{\Des}{\mathrm{Des}}
\newcommand{\RA}{R^{\Anc}}
\newcommand{\RD}{R^{\Des}}
\newcommand{\vA}{v^{\Anc}}
\newcommand{\vD}{v^{\Des}}
\newcommand{\fringe}{\mathrm{fringe}}
\newcommand{\ring}{{\mathrm{ring}}}
\newcommand{\congest}{\textsf{CONGEST}}
\renewcommand{\forall}{\mathrm{\text{ for all }}}
\renewcommand{\top}{\mathrm{top}}
\renewcommand{\bot}{\mathrm{bot}}

\newcommand{\congestion}{{\sf congestion}}
\newcommand{\dilation}{{\sf dilation}}
\newcommand{\one}{\mathbbm{1}}
\newif\ifrandom
\randomtrue

\newcommand{\defeq}{\stackrel{\mathrm{\scriptscriptstyle def}}{=}}

\newcommand{\poly}{{\rm poly}}

\newcommand{\todolater}[1]{}%{\textcolor{red}{\small {\textbf{(TODO later: }#1\textbf{) }}}}

\author{Arun Jambulapati \\
Stanford University	\\
\texttt{jmblpati@stanford.edu}
	 \and 
Yang P. Liu \\
Stanford University \\
\texttt{yangpatil@gmail.com}
\thanks{Research supported by the U.S.
Department of Defense via an NDSEG fellowship.} 
\and 
Aaron Sidford \\
Stanford University \\
\texttt{sidford@stanford.edu}
\thanks{Research supported by NSF CAREER Award CCF-1844855.}
}

\begin{document}
%\pagenumbering{gobble}

\title{Parallel~Reachability~in~Almost~Linear~Work~and~Square~Root~Depth}

%\title{Almost Optimal \sidford{am hesitant to use optimal since we don't know.} Width Parallel Algorithms for Directed Reachability in Nearly Linear Work}

\begin{titlepage}
\clearpage\maketitle
\thispagestyle{empty}

\begin{abstract}
	
In this paper we provide a parallel algorithm that given any $n$-node $m$-edge directed graph and source vertex $s$ computes all vertices reachable from $s$ with $\otilde(m)$ work and $n^{1/2 + o(1)}$ depth with high probability in $n$. This algorithm also computes a set of $\otilde(n)$ edges which when added to the graph preserves reachability and ensures that the diameter of the resulting graph is at most $n^{1/2 + o(1)}$. Our result improves upon the previous best known almost linear work reachability algorithm due to Fineman \cite{Fine18} which had depth $\O(n^{2/3})$. 

Further, we show how to leverage this algorithm to achieve improved distributed algorithms for single source reachability in the \congest~model. In particular, we provide a distributed algorithm that given a $n$-node digraph of undirected hop-diameter $D$ solves the single source reachability problem with $\otilde(n^{1/2} + n^{1/3 + o(1)} D^{2/3})$ rounds of the communication in the \congest~model with high probability in $n$. Our algorithm is nearly optimal whenever $D = O(n^{1/4 - \epsilon})$ for any constant $\epsilon > 0$ and is the first nearly optimal algorithm for general graphs whose diameter is $\Omega(n^\delta)$ for any constant $\delta$. 
 
 \todolater{TODO: Many of the theorems are ``expectation with high probability", think about to write this in a readable way especially in \cref{sec:parallel}}
\end{abstract}

\end{titlepage}

\newpage

\section{Introduction}

Given a $n$-vertex $m$-edge \emph{directed graph} or \emph{digraph} $G = (V,E)$ and a vertex $s \in V$ the \emph{single source reachability problem} asks for the set of vertices $T \subseteq V$ reachable from $s$, i.e. the vertices $t \in V$ for which there is a $s$ to $t$ path in $G$. This problem is perhaps one of the simplest graph optimization problems. It easily solvable in linear, $O(n + m)$, time by any of a number of classic graph exploration algorithms, e.g. breadth first search (BFS), depth first search (DFS), etc. and is often one of the first graph problems considered in an introductory algorithms course. The reachability problem is prevalent in theory and practice and algorithms for solving it are  leveraged to solve more complex graph optimization problems, including computing strongly connected components, shortest paths, maximum flow, spanning arborescenses, etc. 

Given the fundamental nature of the reachability problem and the utility of reachability algorithms, the reachability problem is often one of the first considered when investigating resource constrained computation. However, despite the simplicity of solving single source reachability with optimal time complexity, obtaining optimal algorithms for this problem under constraints of parallelism \cite{Spencer1997,UY91,Fine18}, distributed computation \cite{Nan14,GU15}, space utilization \cite{Wigderson92,BarnesBRS98}, and dynamic updates \cite{AbboudW14} are all notoriously difficult. In many cases, reachability lies at the heart of well-known long-standing open problems in complexity theory. For example, reachability is known to be complete for non-deterministic log space (NL) computation \cite{Wigderson92} and obtaining sufficiently efficient dynamic reachability algorithms would break the popular  3-SUM conjecture and refute the strong exponential time hypothesis (SETH) \cite{AbboudW14}. 

In the parallel and distributing models of computation, single-source reachability is itself a fundamental barrier towards achieving efficient graph optimization.  Despite extensive study, until a recent breakthrough of Fineman \cite{Fine18} the best parallel reachability algorithm all required trading off depth versus work and all known algorithms that had linear work had the trivial $O(n)$ depth. In distributed computation, for example the popular \congest-model  \cite{Pel00}, though there have been  algorithmic improvements over the trivial $O(n)$ round protocol \cite{Nan14,GU15}, the best known algorithms are polynomial factors larger then the known $\Omega(D + \sqrt{n})$-round lower bound (where here $D$ denotes the undirected hop-diameter of the graph) \cite{DHK11}.
 
Given these complexity theoretic barriers related to reachability and the prevalance of parallel and distributed models of computation, improved parallel and distributed reachability algorithms are highly coveted.  In this paper we provide improved reachability algorithms under each computational model. The main results of this paper are as follows.

\begin{theorem}[Parallel Reachability]
\label{thm:main_parallel}
\label{thm:parallel_informal}
There is a parallel algorithm that given a $n$-node $m$-edge digraph solves the single source reachability problem with work $\O(m)$  and depth $n^{1/2 + o(1)}$ with high probability in $n$.
\end{theorem}

\begin{theorem}[Distributed Reachability]
\label{thm:main_distributed}
\label{thm:congest}
\label{thm:distributed_informal}
There is a distributed algorithm that given a $n$-node digraph of undirected hop-diameter $D$ solves the single source reachability problem with $\O(n^{1/2} + n^{1/3 + o(1)} D^{2/3})$ rounds of the communication in the \congest~model with high probability in $n$.
\end{theorem} 

\Cref{thm:parallel_informal} improves the previous best $\otilde(n^{2/3})$ depth bound achieved by a parallel nearly linear work algorithm due to Fineman \cite{Fine18}. \Cref{thm:distributed_informal}  is nearly optimal whenever $D = O(n^{1/4-\eps})$ for some $\eps > 0$ due to a known $\tilde{\Omega}(\sqrt{n} + D)$ lower bound \cite{DHK11} and is the first nearly optimal algorithm for general directed graphs where $D = \Omega(n^\delta)$ for for constant $\delta > 0$. (See \Cref{sec:related_work} for a more detailed comparison to and discussion of previous work.)

Our results build upon a recent breakthrough result of Fineman \cite{Fine18} and a simple, yet-powerful decompositional tool regarding reachability known as \emph{hopsets} or \emph{shortcuts}. Shortcuts are edges which if added to the graph, do not change which pairs of vertices can reach each other. It is well known that in a graph of diameter $D$, i.e. largest shortest-path distance between a pair of vertices which can reach each other is $D$, we can compute reachability in work $O(m)$ and depth $O(D)$. Consequently, a natural approach towards improved reachability algorithms would simply be to find a small set of shortcuts in nearly linear time which decrease the diameter of a graph.

Computing shortcutters is a tantalizing approach to improved reachability algorithms. A simple folklore random sampling argument can be used to show that for every $n$-node digraph and every parameter $t$ there exists a set of $O(t^2 \log^2 n)$ shortcuts such that adding them makes the diameter of the graph is at most $O(n/t)$. Consequently, there is a nearly linear, $O(n \log^2 n)$, number of edges which would reduce the diameter to $O(\sqrt{n})$. Finding such a set of shortcuts in nearly linear work and $O(\sqrt{n})$ depth would immediately yield linear work $O(\sqrt{n})$ depth algorithms for reachability. Unfortunately, even obtaining almost linear time algorithms for constructing a set of shortcutters which reduce the diameter to almost square root $n$ was open prior to this work. 

Fortunately, recent work of Fineman \cite{Fine18} provided some hope towards achieving this goal. This work provided the first nearly linear time algorithm for computing a nearly linear number of shortcutters which provide any polynomial diameter reduction from the trivial $O(n)$ bound. In particular, Fineman's work provided a nearly linear time algorithm which computed with high probability a nearly linear number of shortcutters that decrease the diameter to $\O(n^{2/3})$ and leveraged this result to obtain a parallel reachability algorithm with $\O(m)$ work and $\O(n^{2/3})$ depth that succeeds with high probability in $n$. 

Though an impressive result and a considerable breakthrough, this work left open the question of how well parallel almost linear work algorithms could match the depth bound that would be optimistically predicted by hopsets, i.e. $O(n^{1/2})$. Further, this work left open the question of whether these improvements could be transferred to additional resource constrained computational problems. In this paper we make progress on both questions with \Cref{thm:parallel_informal} and \Cref{thm:distributed_informal}. We provide an $\otilde(m)$ work and $n^{1/2 + o(1)}$ depth algorithm that computes a set of $\otilde(n)$ shortcutters that reduces the diameter to $n^{1/2 + o(1)} $ and use this to achieve improved distributed algorithms.

We achieve our results by strengthening and simplifying parts of Fineman's algorithm (see \Cref{sec:overview} for an overview of the approach). Further, we provide a fairly general strategy to turn improved parallel reachability algorithms into improved distributed algorithms in the \congest-model, building off approaches of \cite{Nan14,GU15}. Interestingly, we show that even Fineman's algorithm can be be modified to achieve improved distributed algorithms (albeit with weaker bounds).

Ultimately we hope this work sheds light on the structure of single source reachability, may lead to faster reachability in more  resource constrained computational environments, and may ultimately lead to more practical massive scale graph processing. 

\paragraph{Paper Outline} The rest of the paper is structured as follows. In the remainder of this introduction we formally state our results in \Cref{sec:our_results} and compare to previous work in \Cref{sec:related_work}. In \Cref{sec:prelim} we cover technical preliminaries and leverage this notation to provide a more technical approach overview in \Cref{sec:overview}. In \Cref{sec:seq} we then provide an $\otilde(m)$ time algorithm for computing $\otilde(n)$ shortcuts which decrease the diameter to $n^{1/2 + o(1)}$ with high-probability. This serial algorithm demonstrates many of the key insights we ultimately build upon to achieve our parallel reachability algorithms in \Cref{sec:parallel} and our distributed reachability algorithms in \Cref{sec:distributed}.

\subsection{Our results}
\label{sec:our_results}

In this paper we provide several parallel and distributed algorithms for efficiently constructing diameter-reducing hopsets and computing reachability in digraphs. Here we provide a brief overview of these results. Throughout this section (and the rest of the paper) we use $\otilde(\cdot)$ to hide polylogarithmic factors in $n$ and we use w.h.p. as shorthand for ``with high probability in $n$" where in both cases $n$ is used to denote the number of vertices in the original input graph. 

First, in \Cref{sec:seq} we provide a sequential algorithm for efficiently computing diameter-reducing hopsets. This algorithm improves upon the previous best diameter bound of $\O(n^{2/3})$ \cite{Fine18}, known to be achievable by a nearly linear time algorithm. Our main result is as follows:

\begin{theorem}[Sequential Diameter Reduction]
\label{thm:seq}
For any parameter $k$, there is an algorithm that given any $n$-node and $m$-edge digraph in $\O(mk)$ time computes  $\O(nk)$ shortcuts such that adding these edges to the graph reduces to the diameter to $n^{1/2 + O\left(1 / \log k\right)}$ w.h.p.
\end{theorem}

This result forms the basis for our improved parallel reachability algorithms. We first argue that careful modification and application of the algorithmic and analytic insights of \Cref{sec:seq} suffice to obtain similar diameter improvements from a parallel algorithm. Our main result of \Cref{sec:parallel} is the following parallel analog of \Cref{thm:seq}.

\begin{theorem}[Parallel Diameter Reduction]
\label{thm:parallel}
For any parameter $k$, there is a parallel algorithm that given any $n$-node and $m$-edge digraph with $\O(mk + nk^2)$ work and  $\poly(k) \cdot n^{1/2+ O(1/\log k)}$ depth computes a set of $\O(nk)$ shortcuts such that adding these edges to the graph reduces the diameter to $n^{1/2 + O(1/\log k)}$ w.h.p.
\end{theorem}

Setting $k = O(\log n)$ immediately gives the following corollary.

\begin{corollary}
There is a parallel algorithm that given any $n$-node and $m$-edge digraph performs $\O(m)$ work in depth $n^{1/2 + o(1)}$ and computes a set of $\O(n)$ shortcuts such that adding these edges to the graph reduces the diameter to $n^{1/2 + o(1)}$ w.h.p.
\end{corollary}

Further, since single source reachability can be solved by BFS in linear work and depth proportional to the diameter of the graph applying this corollary and then leveraging BFS immediately proves \Cref{thm:main_parallel}, our main result on a parallel solution to single source reachability.

We leverage this parallel reachability result to provide our improved distributed reachability algorithms in \Cref{sec:distributed}. Formally we consider the \congest-model where given a $n$-node digraph $G$  there is a separate processor for each node and in every round, for every vertex $u$, its processor may send $O(\log n)$-bits of information to each of its neighbors (i.e. vertices $v$ for which either $(u,v)$ or $(v,u)$ is an edge). Here the distributed reachability question we consider is how to design a messaging scheme so that each node learns if it is reachable from a single given source in as few rounds as possible. 

Our main result regarding such distributed algorithms is given by \cref{thm:main_distributed}. Formally, we show that if $D$ is the diameter of the undirected graph associated with $G$ (i.e. there is an edge in the undirected graph between $u$ and $v$ if and only if either $(u,v)$ of $(v,u)$ is an edge in $G$) then we can design a distributed algorithm that solves single source reachability in $\otilde(n^{1/2} + n^{1/3 + o(1)} D^{2/3})$ rounds of communication in the \congest-model with high probability in $n$. Due to a known lower bound of $\tilde{\Omega}(\sqrt{n} + D)$ \cite{DHK11} this result is nearly optimal whenever  $D = O(n^{1/4-\eps})$ for any $\eps > 0$. Further, to the best of our knowledge this is the first nearly optimal algorithm for general directed graphs when $D = \Omega(n^\delta)$ for any constant $\delta \in (0,1/4)$. 

Interestingly, we achieve this improved distributed algorithm by following a fairly general framework inspired by \cite{GU15,Nan14}. We argue that there is a fairly general procedure for converting ``nice enough'' work-efficient parallel reachability algorithms into improved bounds on distributed reachability. This procedure  first computes a set of shortcuts using known prior work on distributed algorithms \cite{Pel00}. Leveraging these shortcuts, the procedure then solves reachability  by applying a work-efficient parallel reachability algorithm over the a graph with shortcuts added. We then argue that if the reachability algorithm is ``nice enough'' we can bound the distributed round complexity of the resulting algorithm as a function of the work and depth of the parallel reachability algorithm. Ultimately, \cref{sec:distributed} shows that in addition to proving \cref{thm:main_distributed} we could have used this framework and Fineman's recent work \cite{Fine18} to obtain improved distributed algorithms (even without using our work of \cref{sec:parallel}); albeit with worse bounds.

\subsection{Related Work}
\label{sec:related_work}

The problem of computing single-source reachability from source $s$ in $n$-node $m$-edge digraph $G$ is  one of the most fundamental questions in the theory of parallel algorithms. A more complete survey of previous results can be found in \cite{Fine18} and we describe them  briefly here. 

Two folklore algorithms exist for parallel reachability. First, the complete transitive closure of $G$ can be computed in parallel by repeatedly squaring the adjacency matrix of $G$. This achieves a polynomial running time with polylogarithmic depth, but the $\otilde(n^\omega)$ work algorithm where $\omega < 2.373$ is the matrix multiplication constant \cite{Williams12} is currently not known to be nearly linear even for dense graphs. Second, a straightforward modification to standard breadth-first search, called ``Parallel BFS,'' enables us to compute the single-source reachability from $s$ in $O(m)$ time and $O(n)$ depth. This has near optimal work, but its depth is trivial, i.e.  it is essentially a fully serial algorithm. 

Procedures by Spencer \cite{Spencer1997} and Ullman and Yannakis \cite{UY91} provide work-depth tradeoffs that interpolate between the extremes of these two naive procedures, yet neither improve upon the $\O(n)$ depth bound in the case of $\otilde(m)$ work. Fineman \cite{Fine18} provided a breakthrough by demonstrating the $\otilde(m)$ work algorithm for single-source reachability achieving sublinear depth $\otilde(n^{2/3})$. Our results build on those of Fineman by improving the depth to $n^{1/2 + o(1)}$. These results are summarized in \Cref{table:comparison}.  

\begin{table}
  \begin{center}
    \small
    \begin{tabular}{|c|c|@{\hspace{1em}}l|}\hline & Work & Span \\ \hline\hline
      Parallel BFS & $O(m)$ & $\otilde(n)$ \\
      \hline Parallel Trans. Closure & $\otilde(n^\omega)$ & $\otilde(1)$ \\ \hline
      Spencer's~\cite{Spencer1997} & $\otilde(m+n\rho^2)$ & $\otilde(n/\rho)$ \\ \hline
      UY~\cite{UY91} & $\otilde(m\rho + \rho^4/n)$ & $\otilde(n/\rho)$ \\ \hline
      Fineman's~\cite{Fine18} & $\otilde(m)$ & $\otilde(n^{2/3})$  \\ \hline
      \textbf{This paper} & $\otilde(m)$ & $n^{1/2 + o(1)}$  \\ \hline
    \end{tabular}\vspace{-1em}
  \end{center}
  \caption{Summary of previous results for parallel single-source digraph reachability. 
  Here $\rho \in [1,n]$ can be chosen arbitrarily and $\omega < 2.373$ denotes the matrix multiplication exponent.
  This table was modified from the one in \cite{Fine18}.
}
  \label{table:comparison}
\end{table}

Our algorithm, like \cite{Fine18}, solves a more general problem than single-source reachability: we show show that our algorithm gives  work-efficient parallel construction of an $\otilde(n)$-edge hopset of diameter $n^{1/2 + o(1)}$. With this hopset we can answer arbitrary single-source reachability queries in nearly-linear work and $n^{1/2 + o(1)}$ depth. A natural question is whether our hopset construction can be improved. However even from the perspective of constructibility the true tradeoff between diameter and number of added edges is not known. As mentioned previously a straightforward random construction provides a $O(t^2 \log^2 n)$-edge hopset with diameter $O(n/t)$, but it is not known how to improve upon this result in any regime.  Building off of \cite{Hesse2003}, \cite{HuangP18} demonstrates that $O(n)$-edge hopsets cannot ensure diameter less than $O(n^{1/6})$ and $O(m)$-edge hopsets cannot achieve $O(n^{1/11})$ diameter; no results are known in other parameter regimes. Nevertheless we conjecture that in the case of nearly-linear sized hopsets our construction is tight and therefore any polynomial improvement to our algorithm must somehow avoid the hopset paradigm. \todolater{Actually, the Huang-Pettie hopset construction can be generalized to the superlinear-shortcut regime. Hopsets of size$ n^{1+a}$ cannot have diameter better than $n^{((d-1)/(d(d+1)) - a/d)}$: this is polynomial for all $a<1$} 

Distributed algorithms in the \congest~model \cite{Pel00} have been studied extensively over the past two decades. Though there have been multiple improvements to the round complexity of approximately solving single source shortest paths in this model (see for example \cite{BeckerKKL17, GhaffariL18,ForsterN18}  for the relevant literature), there has been comparatively little progress on the solving the same problem on directed graphs \cite{Nan14,GhaffariL18,ForsterN18}. For the  single-source reachability problem considered in this paper the previous state of the art for this problem is due to \cite{GU15}, which solved the problem in  $\otilde(D + \sqrt{n} D^{1/4} )$ rounds w.h.p. This algorithm in turn improved upon the $\otilde(D + \sqrt{n} D^{1/2})$ w.h.p. round bound of \cite{Nan14}, which was (to the best of our knowledge) the first non-trivial distributed algorithm for this problem.

\section{Preliminaries}
\label{sec:prelim}

We denote vertex set of a graph $G$ by $V(G)$, and the edge set by $E(G)$. We simply write these as $V$ and $E$ when the graph $G$ is clear from context. For a $V' \subseteq V$, we let $G[V']$ denote the induced subgraph on $V'$, i.e. the graph with vertices $V'$ and edges of $G$ that have both endpoints in $V'$.

\paragraph{Digraph relations:} Let $G$ be a \emph{directed graph} or \emph{digraph} for short. We say that $u \pe v$ if there is a directed path from $u$ to $v$ in $G$. In this case we say that \emph{$u$ can reach $v$}, or that \emph{$v$ is reachable from $u$}. We say that $u \not\pe v$ if there is no directed path from $u$ to $v$ in $G$. In this case we say that $u$ cannot reach $v$, or that $v$ is not reachable from $u$. When $u \pe v$ and $v \pe u$ we say that $u$ and $v$ are in the same strongly connected component. We define the \emph{descendants} of $v$ to be $\RD_G(v) \defeq \{ u \in V(G) : v \pe u\}$ and the \emph{ancestors} of $v$ to be $\RA_G(v) \defeq \{u \in V(G) : u \pe v \}.$ We say that $u$ and $v$ are \emph{related} if $u \pe v$ or $v \pe u$. We define the \emph{related vertices} of $v$ as $R_G(v) \defeq \RD_G(v) \cup \RA_G(v).$ Throughout, the letter $R$ we use in the notation should be read as ``related" or ``reachable". We say that $u$ is \emph{unrelated} to vertex $v$ if  $u \in V(G) \bs R_G(v)$. 

We extend this notation to subsets $V' \subseteq V$ in the natural way. We define the ancestors, descendants, and related vertices to $V'$ as
\[ \RD_G(V') \defeq \bigcup_{v \in V'} \RD_G(v) \text{ , } \RA_G(V') \defeq \bigcup_{v \in V'} \RA_G(v) \text{ , and } R_G(V') \defeq \RD_G(V') \cup \RA_G(V'). \] We say that a vertex $v$ is related to a subset $V'$ if $v \in R_G(V').$ When the graph $G$ is clear from context, we will often drop the $G$ subscript and simply write (for example) $\RD(v), \RA(v),$ and $R(v)$ instead of $\RD_G(v), \RA_G(v), R_G(v).$

We further extend this notation to induced subgraphs of $G$. Let $G'$ be a subgraph of $G$, possibly with a different set of vertices and edges than $G$. We say that $u \pe_{G'} v$ if there is a directed path from $u$ to $v$ in the subgraph $G'$; we say that $v$ is reachable from $u$ \emph{through $G'$} in this case. Define
 \[ \RD_{G'}(v) \defeq \{ u \in V(G') : v \pe_{G'} u \}, \RA_{G'}(v) \defeq \{ u \in V(G') : u \pe_{G'} v \}, \text{and}
 \] 
$R_{G'}(v) \defeq \RD_{G'}(v) \cup \RA_{G'}(v)$. We similarly extend this definition to subsets of vertices $V' \subseteq V(G')$:
\[ \RD_{G'}(V') \defeq \cup_{v \in V'} \RD_{G'}(v), \RA_{G'}(V') \defeq \cup_{v \in V'} \RA_{G'}(v), \text{ and } R_{G'}(V') \defeq \RA_{G'}(V') \cup \RD_{G'}(V'). \] As our algorithm performs recursion on subgraphs of $G$, this notation enables us to reference specific subproblems as our algorithm progresses.

A \emph{shortcut} refers to adding an edge $(u,v)$ to a graph $G$ where $u \pe v$ in $G$. Adding the edge does not affect the reachability structure of $G$. A \emph{shortcutter} $v$ is a node we add shortcut edges to and from.
%: our algorithm will choose several shortcutter nodes and add many shortcut edges both to and from them.
A \emph{hopset} refers to a collection of shortcuts. 

\paragraph{Paths:} Our analysis will consider paths in the graph as well as the relations between the vertices on the path and other vertices in the graph. Let $G$ be a digraph. We denote a path $P = \l v_0, v_1, \dots, v_\ell \r$, where all the $v_i$ are vertices of $G$ and $(v_i,v_{i+1}) \in E(G)$. Here, the length of the path is $\ell$, where we have that $v_0 \pe v_1 \pe \dots \pe v_\ell$. We say that the head of the path is $\head(P) \defeq v_0$ and the tail is $\tail(P) \defeq v_\ell$. We now make the following definitions.
	
\paragraph{Path-related vertices:} We adopt a similar convention as \cite{Fine18}. For a path $P = \l v_0, v_1, \dots, v_\ell \r$ we say that $v$ is \emph{path-related} if $v \in R_G(P)$. Further, for any path $P$ in digraph $G$, we define $s(P, G) \defeq |R_G(P)|$ as the number of path-related vertices. All path-related vertices are one of the following three types:
\begin{itemize}
\item \textbf{Descendants:} We say that a vertex $v$ is a descendant of the path $P$ if $v \in \RD_G(P)\bs \RA_G(P).$ Note that this holds if and only if $v_0 \pe v$ and $v \not\pe v_\ell.$
\item \textbf{Ancestors:} We say that a vertex $v$ is an ancestor of the path $P$ if $v \in \RA_G(P)\bs \RD_G(P).$ Note that this holds if and only if $v_0 \not\pe v$ and  $v \pe v_\ell.$
\item \textbf{Bridges:} We say that a vertex $v$ is a bridge of the path $P$ if $v \in \RD_G(P)\cap \RA_G(P).$ Note that this holds if and only if $v_0 \pe v$ and $v \pe v_\ell.$
\end{itemize}
A vertex which is not a descendant, ancestor, or bridge for a path $P$ is called \emph{unrelated} to $P$.
Later in \cref{sec:parallel} we explain how to extend all these definitions to the distance-limited case.

\paragraph{Subproblems:} During our algorithms' recursions, we will make reference to the induced recursive calls made. Consider a graph $G$ and a path $P$ in $G$. During a call to an algorithm on the graph $G$, we define a \emph{subproblem} to be an induced subgraph $G[V']$ along with a subpath $P'$ of $P$ which lies inside $G[V']$ on which we perform a recursive execution.

\paragraph{Miscellaneous:} We let $B(n,p)$ be the binomial random variables over $n$ events of probability $p$. We have the following standard fact about binomial random variables:

\begin{lemma}[Chernoff Bound]
\label{lemma:chernoff}
Let $X \sim B(n,p)$ be a binomial random variable. Then
\[
\Pr\left[X > (1+\delta)np\right] \leq \exp\left(-\frac{\delta^2}{2+\delta} np\right). 
\]
\end{lemma}

The diameter of a directed graph $G$ is defined as $\max \{ d(u, v) : u, v \in V(G) \text{ and } d(u, v) < \infty \}$ i.e. the longest shortest path between two vertices $u, v$ where $u$ can reach $v$.

For functions $f(n)$ and $g(n)$ we say that $f(n) = \O(g(n))$ if $f(n) = O(g(n) \cdot \poly\log n)$. In particular, $\O(1) = O(\poly \log n).$

\section{Overview of Approach}
\label{sec:overview}

Here we provide an overview of our approach towards achieving our algorithmic improvements on reachability. First, we provide a blueprint for our sequential nearly linear time algorithm for computing diameter reducing hopsets. While further work is needed to make this algorithm implementable in low depth and more insights are needed to obtain our parallel and distributed results (and we discuss these briefly), we believe this simple sequential algorithm demonstrates the primary algorithmic insights of the paper.

For the remainder of this section, let $G = (V,E)$ be a digraph for which we wish to efficiently compute diameter-reducing shortcuts. For simplicity we consider the case where $G$ is a directed acyclic graph (DAG); the analysis of the general case is essentially identical and for the purposes of reachability (ignoring parallel computation issues) we can contract every strongly connected component to a single vertex. 

Our shortcutting algorithms follows the general blueprint leveraged by Fineman \cite{Fine18} for efficiently computing diameter reducing shortcuts. Briefly, Fineman's algorithm consists of the following iteration: in every step, a ``shortcutter" vertex $v$ is selected from $V$ uniformly at random. It then constructs three sets: $v$'s ancestors $\RA_G(v)$, $v$'s descendants denoted $\RD_G(v)$, and the set of notes \emph{unrelated to $v$} $U_G(v) \defeq V \bs \{\RD_G(v) \cup \RA_G(v) \}$. The algorithm then adds shortcut edges from $v$ to every node in $\RD_G(v)$ and from every node in $\RA_G(v)$ to $v$. The algorithm then computes the induced graphs $G_D =  G[\RD_G(v)], G_A = G[\RA_G(v)]$ and $G_U = G[U_G(v)]$ and recursively applies the procedure to each of these three graphs. 

To analyze this procedure, consider any path $P$ in $G$. The algorithm in \cite{Fine18} considers how the shortcuts the algorithm constructs affect the distance between the endpoints of $P$. When a shortcutter vertex $v$ is picked, there are four possibilities with how it interacts with $P$: it is either 
\begin{enumerate}
    \item Unrelated to every node in $P$.
    \item An ancestor to some nodes in $P$ forming a subpath $P_1$ and unrelated to the remaining subpath $P_2$.
    \item A descendant to some nodes in $P$ forming a subpath $P_2$ and unrelated to the remaining subpath $P_1$.
    \item An ancestor to the tail of  $P$ and a descendant to the head of $P$.
\end{enumerate}
Consider shortcutting through any vertex $v$ and following the recursion of the above algorithm. By the above, it is clear that after shortcutting through any vertex $v$ one of three things can happen: either $P$ remains intact in a subproblem (case 1), it gets split into exactly two pieces in two different subproblems (cases 2 or 3), or the connectivity between the endpoints of $P$ is resolved through $v$, i.e. we can go from $P$'s head to tail in two edges by going through $v$ (case 4). Thus we either split $P$ into at most two pieces or we ensure the endpoints $P$ are distance $2$ from each other. Let $P_i$ be the pieces $P$ is split into at some state of the algorithm's execution, and let $V_i$ be the subproblem vertex set containing $P_i$. 

The key insight of Fineman is to define the following function (which we defined in \cref{sec:prelim}) and to use it to reason about the effect of this random process:
\[ s(P_i, G_i) \defeq  \left|\{u \in V(G_i)| u \text{ is related to some node in } P_i\}\right|. \] 
Observe that $s(P_i, G_i)$ is an overestimate of the length of $P_i$, and that $s(P, G) \le n$. Define 
\[ L(P) \defeq \sum_{i} s(P_i, G_i) \]
to be the sum over all $s(P_i, G_i)$ at this state of our algorithm. $L(P)$ is a random variable, but by reasoning about the above cases we can reason about how $L(P)$ changes in expectation. For any subpath $P_i$, consider the induced subproblems after shortcutting through a randomly selected node $v$. If $v$ lands in case 1 nothing changes, and if $v$ lands in case 4 we resolve the connectivity of $P_i$ and set $s(P_i, G_i)$ to $0$, as there is no remaining subproblem. In cases $2$ and $3$ however, we split $P_i$ into two pieces: call these $P_{i1}$ and $P_{i2}$. In these cases, Fineman is able to argue that a randomly chosen node can ensure that the number of nodes which are related to either $P_{i1}$ or $P_{i2}$ decreases by some constant factor $c$ in expectation. Thus if $f(x)$ is the expected shortcut length of a path $P_i$ with $s(P_i, G_i) = x$ we can essentially guarantee by induction that $f(x) \le \max_{a+b = cx} f(a) + f(b)$. Fineman achieves a constant $c = 3/4$: this gives a bound of $f(n) \le O(n^{1/\log(8/3)})$. A sophisticated refinement of this argument allows him to obtain his claimed $\O(n^{2/3})$ bound.

Our algorithm is almost identical to that of Fineman with one crucial modification: we pick more than one shortcutter node before we recurse. Specifically, we shortcut from $k$ random vertices in the graph instead of only a single vertex. After shortcutting, we partition the vertices of the graph into subsets, much like Fineman's algorithm partitioned the vertices in ancestors, descendants, and unrelated vertices. In our partitioning scheme, two vertices are in the same subset if and only if they have the same relationship to each of the $k$ shortcutters. As an example, two vertices $u_1$ and $u_2$ are \emph{not} in the same subset if say $u_1$ is an ancestor of shortcutter $v$ and $u_2$ is a descendant of shortcutter $v$. If we pick $k$ shortcutter nodes from cases $2$, $3$, or $4$ at a time and partition in the way described, we are able to guarantee that the number of path-related nodes after we recurse decreases by a factor of $\frac{2}{k+1}$ in expectation after we recurse (\cref{seq:2k+1lemma}). Although the path splits into $k+1$ pieces after recursing, analyzing the resulting recursion in the same manner as Fineman reveals that $k = \omega(1)$ will ensure our algorithm will shortcut paths to length $n^{1/2+ o(1)}$ as desired. Unfortunately, we are not able to guarantee this directly. The above analysis requires that in any recursive level we pick either $0$ or $k$ path-relevant shortcutters in any recursive level; however we do not know how to obtain such fine-grained control without knowing the path.

Intuitively, we would like to pick as many shortcutter vertices as possible while staying within our nearly-linear work bound-- the more shortcutters we pick, the more likely we are to obtain the $\frac{2}{k+1}$ reduction in path-related nodes. However, we cannot simply pick the same number of shortcutters in every level of recursion: because the number of path-related nodes goes down rapidly, picking $k$ shortcutters per level of recursion will eventually only enable us to pick a single path-relevant shortcutter per round. Instead, we show that after each level of recursion the structure of the subproblems is such that we can pick \emph{$k$ times more} shortcutters while still having nearly-linear work. This, combined with a new inductive analysis in \cref{seq:diambound} to get around the fact that we don't have as precise control over the change in $L(P)$ enables us to obtain our result.

\paragraph{Parallel Implementation:} Our techniques as described give us a nearly-linear work algorithm which constructs a nearly-linear number of shortcuts that reduce the diameter to $n^{1/2+o(1)}$. We make our construction parallel in a similar fashion to Fineman. The key insight to \cite{Fine18}'s parallelization is to consider \emph{$D^{\text{search}}$-restricted} searches; instead of computing the ancestor, descendant, and unrelated sets with full graph traversals from a vertex $v$, Fineman computes collections of $D^{\text{search}}$-ancestors and $D^{\text{search}}$-descendants. These are the set of ancestors (resp. descendants) which are reachable from $v$ at distance at most $D^{\text{search}}$. Now although these can be computed in low depth, we cannot use these as a direct replacement for the full ancestor and descendant sets as we can no longer guarantee an expected decrease in $L(P)$.

Fineman gets around this issue with a new idea. Let $G$ be a digraph. Assume that we could efficiently find a set of edges $F$ to add to $G$ such that if $s$ and $t$ are nodes at distance $D$ from each other their distance in $G \cup F$ is at most $D/5$ w.h.p. Then for any nodes $u, v$ at distance more than $D$ from each other we observe that their distance in $G \cup F$ is halved w.h.p., i.e. this breaks up the $u-v$ shortest path into chunks of length $D$ and observe that each subpath's size falls by a constant factor with constant probability. Thus by repeating this procedure on $G \cup F$ and iterating $O(\log n)$ times we observe that every pair of reachable nodes $u,v$ can be brought within distance $D$. Doing this reduction only costs logarithmic factors in total work and parallel depth.

Fineman therefore modifies his recursion in the following way. In every level set $D^{\text{search}} = (\kappa+1) D$, where $\kappa$ is a random variable. Fineman then constructs the $D^{\text{search}}$-ancestors and $D^{\text{search}}$-descendants, but he then defines the unrelated set to be set of all nodes which are \emph{not $\kappa D$-ancestors or $\kappa D$-descendants}. This modification duplicates all nodes at distance between $\kappa D$ and $(\kappa + 1)D$ from the shortcutter $v$, but now any path of length $D$ is partitioned into two contiguous subpaths, copies of which can be found in these three induced sets. Fineman argues that the expected increase in the number of nodes can be controlled and that an analogous bound on $L(P)$ in expectation can be obtained as in the serial setting by picking and choosing the specific copies of subpaths to split $P$ into in the recursion. Combining these pieces allows him to obtain his claimed $\O(n^{2/3})$ depth algorithm. Our approach (\cref{sec:parallel}) will build off of these ideas with several modifications for our new algorithm to leverage our inductive analysis.

\paragraph{Application to Distributed Reachability:} Let $G$ be a $n$-vertex directed graph with undirected hop diameter $D$. We describe our approach for solving the single source reachability problem in the \congest~model on $G$. Our approach involves combining our parallel diameter reduction algorithm with the approaches of Ghaffari and Udwani \cite{GU15} and Nanongkai \cite{Nan14}. The approach (loosely) involves using $\O(\alpha + n/\alpha)$ rounds of communication in the \congest~model to reduce the problem to computing reachability on a set of vertices $S$ of size $\alpha$, with the difference that the vertices must communicate via global broadcasting, as the vertices in $S$ aren't actually connected in the original graph. We then simulate our parallel reachability algorithm on $S$. By analyzing our parallel reachability algorithm, we can analogously get a bound on the number of rounds needed to simulate it in the \congest~model.

\section{Sequential Algorithm}
\label{sec:seq}

The main goal of this section is to prove \cref{thm:seq} showing that for all $k$ there is an $\otilde(mk)$ time algorithm which adds $\otilde(nk)$ shortcuts which reduce the diameter to $n^{1/2+ O(1 /  \log k)}$ w.h.p. In \cref{sec:seq_alg_desc} we present our algorithm for achieving this result. In \cref{sec:seq_work_and_shortcut} we bound the work of the algorithm and the number of shortcuts it adds. In \cref{seq:path_related} we provide our main technical lemma regarding diameter reduction and then in  \cref{sec:recurse_and_diam} we apply this lemma repeatedly to prove that the algorithm reduces diameter, thereby proving \cref{thm:seq}. 

\subsection{Algorithm Description }
\label{sec:seq_alg_desc}

Here we present our sequential short-cutting algorithm (see \cref{algo:seq}). Before stating the algorithm, we give some definitions and intuition for the quantities defined in the algorithm. Let $G$ be the graph that we input to our algorithm and consider the following.
\begin{itemize}
	\item \textbf{Inputs $k, r$}: $k$ is a parameter governing the speed that we recurse at. Intuitively, our algorithm  picks shortcutters so that graphs at one level deeper in the recursion are ``smaller" by a factor of $k$. This is made precise in \cref{seq:rvbound}. $r \le \log_k n$ is the depth of recursion that the algorithm is currently at, where we start at $r = 0$.
	\item \textbf{Set $S$}: $S$ is the set of vertices from which we search and build shortcuts from.
	\item \textbf{Set $F$}: $F$ is the final set of shortcuts we construct.
	\item \textbf{Probability $p_r$}: At recursion depth $r$, for each vertex $v \in V(G)$, we put $v$ in $S$ with probability $p_r$.
	\item \textbf{Labels $\vD, \vA, \xmark$}: We want to distinguish vertices by their relations to vertices in $S$. Therefore, when we search from a vertex $v$ we add a label $\vD$ to add vertices in $\RD_G(v)\bs \RA_G(v)$, a label $\vA$ to all vertices in $\RA_G(v)\bs \RD_G(v)$, and a label $\xmark$ to all vertices in $\RD_G(v) \cap \RA_G(v).$ The label $\xmark$ should be understood as ``eliminating" the vertex (since it is in the same strongly connected component as $v$ and we have shortcut through $v$ already).
\end{itemize} 

\begin{algorithm}[h!]
\caption{$\Seq(G, k, r)$. Takes a graph $G$, parameter $k$ and recursion depth $r \le \log_k n$ (starts at $r = 0$). Returns a set of shortcut edges to add to $G$. Sequential diamater reduction algorithm. $n$ denotes the number of vertices at the top level of recursion.}
\begin{algorithmic}[1]
\State $p_r \assign \min\left\{1, \frac{20 k^{r+1} \log n}{n} \right\};$ \Comment{Begin level $r$ of recursion.}
\State $S \assign \emptyset;$
\For{$v \in V$}
	\State With probability $p_r$ do $S \assign S \cup \{v\};$
\EndFor
\State $F \assign \emptyset;$
\For{$v \in S$} \label{line6}
	\For{$w \in \RD_G(v)$} add edge $(v,w)$ to $F$; \label{line7}
	\EndFor
	\For{$w \in \RA_G(v)$} add edge $(w,v)$ to $F$. \label{line8}
	\EndFor
	\For{$w \in \RD_G(v)\bs \RA_G(v)$} add label $\vD$ to vertex $w$. \label{line9}
	\EndFor
	\For{$w \in \RA_G(v)\bs \RD_G(v)$} add label $\vA$ to vertex $w$. \label{line10}
	\EndFor
	\For{$w \in \RD_G(v)\cap \RA_G(v)$} add label $\xmark$ to vertex $w$. \label{line11}
	\EndFor
\EndFor
\State $W \assign \{ v \in V : v \text{ has no label of } \xmark \}.$ \label{line13}
\State $V_1, V_2, \dots, V_\ell \assign$ partition of $W$ such that $x, y \in V_i$ if and only if $x$ and $y$ have the same exact labels. \Comment{Vertices in the $V_i$ have no label of $\xmark$.}
\For{$1 \le i \le \ell$}
	\State $F \assign F \cup \Seq(G[V_i], k, r+1)$ \label{line15}
\EndFor
\State \Return{$F$}
\label{algo:seq}
\end{algorithmic}
\end{algorithm}

Our algorithm can be thought of as an extension of Fineman's shortcut construction procedure. In every iteration, we seek to add as many shortcutters as possible while staying within our claimed work bound. Thus, in the first iteration we add $\otilde(k)$ shortcutters w.h.p. and perform $\otilde(mk)$ work. We then partition the nodes into clusters such that any two nodes $x$ and $y$ which are in the same $V_i$ have exactly the same labels assigned to them by the shortcutters, none of which are $\xmark$. We will show how to implement this step later (\cref{seq:workbound}). We then recursively apply the algorithm within each cluster with a sampling probability that is a factor of $k$ larger. We will show two things. First, we show that the increase in sampling probability is offset by a decrease in the number of related pairs such that the work done in an iteration is the same w.h.p. (\cref{seq:rvbound} and \cref{seq:labelbound}). Second, we show that if we pick $q$ shortcutters that are path related to a path $P$ we get an expected decrease in the number of path-related nodes to all the induced subproblems of $P$ (\cref{seq:2k+1lemma}). This second fact enables us to replace the recursion in Fineman with one that decreases more quickly (\cref{seq:diambound}): this gives our depth improvement. 

\subsection{Work and Shortcut Bound}
\label{sec:seq_work_and_shortcut}
 
In this section we bound the work and number of shortcuts added by \cref{algo:seq}. In any recursion level of our algorithm there are two sources of work. The first source is from computing the requisite labels $\vD, \vA,$ and $\xmark$ for every node $v$ we shortcut from. The second source comes from grouping the nodes by these labels to generate the subproblems for the next level. We will bound both of these sources of work by using a useful fact on the number of ancestors and descendants a node has in the subproblem it belongs to in any level. 
\begin{lemma}
\label{seq:rvbound}
Consider an execution of $\Seq(G, k, 0)$ on $n$-node $m$-edge $G$. With probability $1 - n^{-10}$ in each recursive execution of  $\Seq(G',k,r)$ in line \ref{line15} of \cref{algo:seq} the following holds 
\[
\RD_{G'}(v) \leq n k ^{-r} ~~~ \text{and} ~~~ \RA_{G'}(v) \leq n k^{-r}
~~~ \text{for all } v \in G'.
\]
\end{lemma}
\begin{proof}
We prove by induction on $r$. Clearly the claim is true for the one recursive call at $r = 0$. We will show that assuming the claim for all recursive calls with $r=j$ the result holds for all $r=j+1$ problems with probability at least $1-n^{-11}$. By applying union bound over all $\otilde(1)$ values of $r$ encountered in the algorithm implies the result.

Assume the result holds for every recursive execution with $r=j$. Let $v \in V$ be any vertex, and let $G'$ be the induced subgraph our algorithm is recursively called on with $r=j+1$ which contains $v$. We prove the claim for $\RD_{G'}(v)$ as the claim for $\RA_{G'}(v)$ follows by a symmetric argument.

Observe that the recursive call $\Seq(G',k,j+1)$ is ultimately called through an execution of $\Seq(H, k, j)$ on some $H \subseteq G$. Let $Q$ be the set of nodes in $V(G')$ which are descendants of $v$ in $H$. Now if $|Q| = \RD_{H}(v)$ is less than $n k^{-r}$ we are done since the induced subgraphs we recurse on only decrease in size. Thus assume $|Q| \geq n k^{-r}$.

Let $Q_1, Q_2, \cdots$ be the strongly connected components of $Q$, and consider any topological order over these subsets of $Q$, where $Q_i$ precedes $Q_j $ whenever a path from $Q_i$ to $Q_j$ exists.  Consider any $x,y \in Q$ where $y$ precedes $x$ in this order. We investigate the random choices in $\Seq(H,k,j)$ that lead to $G'$'s formation. Observe that if we chose $y$ as a shortcutter for $H$, $G'$ would not contain $x$ since $v$ is in $G'$ yet $x$ and $v$ receive different labels from $y$: $v$ is $y$'s ancestor but $x$ is either a descendant of or unrelated to $y$. Further, we observe that if we shortcut from $y$ any node $z$ in $y$'s strongly connected component would also fail to be in $G'$: $z$ would be given an $\xmark$ label.  Thus if we shortcut the graph with any of the $n k^{-j}$ nodes which are earliest in the topological order of $Q$ (which are closest to $v$ in $H$) we can guarantee that $v$ in the $G'$ level has at most $n k^{-j}$ descendants. Since we choose each node with probability $\frac{20 k^j \log n}{n}$ we fail to do this with probability at most 
\[
\Big( 1 - \frac{20 k^j \log n}{n} \Big)^{n k^{-j}} \leq e^{-20 \log n} = n^{-20}. 
\]
By union bounding over all vertices in $G'$ and over all induced subgraphs encountered at level $r = j$ we see that our bound holds for all recursive calls with $r=j$ with probability at least $1-n^{-11}$. The result follows.
\end{proof}

We now bound the number of labels any vertex $v$ receives in any recursive execution which contains it. This will provide us with an elegant way to bound the total work of our procedure.

\begin{lemma}
Consider an execution of $\Seq(G, k, 0)$ on $n$-node $m$-edge $G$ with $k \geq 2$. With probability $1 - 2n^{-10}$, every recursive execution $\Seq(G[V_i],k,r)$ assigns at most $80 k \log n$ labels to every node $w \in V_i$ in lines \ref{line9}, \ref{line10}, and \ref{line11}, where $\xmark$ labels assigned by different shortcutters are counted as distinct labels. 
\label{seq:labelbound}
\end{lemma}
\begin{proof}
Note that $v$ receives a label from a shortcutter $u$ only if $u$ is related to $v$. By \cref{seq:rvbound} we have that at most $2 n k^{-r}$ nodes are related to $v$ for all $v$ in all executions of $\Seq$ with probability $1-n^{-10}$. Since we pick nodes in the $r^{th}$ level with probability $p_r = \frac{20 k^{r+1} \log n}{n}$ we see that the probability that more than $80k\log n$ labels are given to $v$, assuming that at most $2nk^{-r}$ vertices are related to $v$, is at most
\[
\Pr \left[B\Big(2 n k^{-r}, \frac{20 k^{r+1} \log n}{n}\Big) > 80k \log n \right] \leq \exp\left(-\frac{40}{3} k \log n\right) \leq n^{-12}
\]
by \cref{lemma:chernoff}. Thus $v$ receives at most $80 k \log n$ labels with probability at least $1 - n^{-12}$. Union bounding this over all nodes in all recursive executions of $\Seq$ implies the result.
\end{proof}

Finally, we conclude this subsection by bounding the total work of $\Seq$, as well as the number of shortcut edges it adds. 

\begin{lemma}
\label{seq:workbound}
Consider an execution of $\Seq(G, k, 0)$ on $n$-node $m$-edge $G$ with $k \geq 2$. With probability $1-2n^{-10}$ $\Seq(G, k, 0)$ runs in $\otilde(mk)$ time and adds $\otilde(nk)$ shortcuts.
\end{lemma}
\begin{proof}
By our given probability of failure, we may assume \cref{seq:rvbound} and \cref{seq:labelbound} hold deterministically. 

We begin by considering a single recursive execution $\Seq(G[V_i], k ,r)$ generated by $\Seq$. We will bound the number of shortcuts added by this call and amount of work it performs before it recurses in \ref{line15}. We will then aggregate these bounds over all recursive executions and obtain our final result. For convenience, let $G[V_i]$ have $\hat{n}$ nodes and $\hat{m}$ edges. 

We first bound the number of shortcuts added by $\Seq(G[V_i], k ,r)$. By \cref{seq:labelbound} we observe that every recursive execution $\Seq(G[V_i], k, r)$ assigns $\otilde(k)$ labels to every $w \in V_i$. As each label corresponds to a shortcut we add in lines \ref{line7} and \ref{line8}, we see that $\Seq(G[V_i], k, r)$ adds $\otilde(k)$ edges to every $w \in V_i$: this is $\otilde(\hat{n}k)$ edges in total. 

We now bound the work performed by $\Seq(G[V_i], k ,r)$. Within a call to $\Seq$, we perform work in two places: within the loop in line \ref{line6} and when generating the partition in line \ref{line13}. We bound the contributions of these sources in order. First, observe that the loop in \ref{line6} can be implemented by computing breadth-first searches forwards and backwards from every $w$ in the shortcutter set $S$. The amount of work needed to apply the labels and and the shortcuts themselves is clearly $\O(\hat{n}k)$ by the above argument, so we need only to bound the cost of running these traversals. 

Observe that by \ref{seq:labelbound} $\Seq(G[V_i], k ,r)$ assigns $\O(k)$ labels to every $w \in V_i$. Now the number of labels $w$ receives is within a factor of two the number of times it is visited in searches. Thus $w$ is visited $\O(k)$ times in our traversals. Each time we encounter $w$ in a traversal we perform a constant amount of work for each edge incident upon it. Thus if $\delta_i(w)$ is the undirected degree of $w$ in $G[V_i]$, the total work performed by $\Seq(G[V_i], k ,r)$ is 
\[ \O\left(\sum_{w \in V_i} k \delta_i(w)\right) = \O(\hm k). \]

We finally bound the cost of generating the partition in line \ref{line13}. We implement this in two parts. First, we check each vertex to see whether it has an $\xmark$ label and discard any vertex which does. Next, we define an order over all possible combinations of labelings a node could receive. We then sort the remaining nodes by this order: we can then trivially read off the partition. To implement this order of labelings, pick an arbitrary ordering on the individual labels we distribute to nodes. To compare two labeling schemes $a$ and $b$ we internally sort $a$ and $b$ by our arbitrary ordering, and then determine the order amongst $a$ and $b$ lexicographically.

To implement this procedure, we first note that by \cref{seq:labelbound} every node receives at most $\otilde(k)$ labels. Determining which of the $\hat{n}$ nodes have an $\xmark$ label clearly takes $O(\hat{n}k)$ time. It is straightforward to verify that comparing two labelings each with at most $\otilde(k)$ labels with this scheme requires $\otilde(k)$ time: thus this partitioning can be found in $\otilde(\hat{n}k)$ time using a mergesort. Combining this with the previous bound we see that $\Seq(G[V_i], k ,r)$ requires $\otilde(\hat{m}k)$ time before recursing. 

We now obtain our final work and shortcut bounds by aggregating. If we consider the set of recursive calls $\Seq(H,k,r)$ for any fixed value of $r$, we see that the calls are applied to a disjoint collection of subgraphs of $G$. Thus, the total number of nodes in all of these subproblems is $n$, and the total number of edges is at most $m$. Thus cost of performing all of these calls without recursing is $\otilde(mk)$, and these calls collectively add $\otilde(nk)$ shortcuts. As there are at most $\otilde(1)$ different values of $r$ our claim follows. 
\end{proof}

\subsection{Path Related Nodes and Main Helper Lemma}
\label{seq:path_related}

We now prove a significant helper lemma that will enable us to prove our diameter bound. We begin with some context. Recall that the goal of our algorithm is, for any path $P \in G$ with endpoints $s$ and $t$, to find a bridge for $P$.  If in a recursive call to $S$ we succeed in finding a bridge for $P$ we add the corresponding shortcuts to connect $s$ and $t$ with a length $2$ path: there is nothing more for us to do. If instead we do not find a bridge in $S$, we observe that $P$ gets split amongst several different node-disjoint subproblems: we then seek to find bridges for each of these subproblems separately. Thus the collection of paths $P_j$ represents the ``residual" paths left for our algorithm to resolve: we either pick a bridge and entirely resolve the path or split it into pieces. While this splitting of the path may seem counterproductive, we show that the total number of path-related vertices in the next recursion level summed over all $P_i$ decreases significantly in expectation when we recurse. We thus can ensure some form of progress whether we resolve the path or not.

In the below lemma, for a path $P'$ in a subgraph $G' \subseteq G$, we define $s(P', G')$ to be the number of vertices in $G'$ that are related to $P'$, as was done in \cref{sec:prelim}.

\begin{lemma}
\label{seq:2k+1lemma}
Let $G$ be a digraph and let $P$ be a path in $G$. Let $T$ be a uniformly random subset of $V[G]$, where any node $v \in V[G]$ is in $T$ with some probability $p$. Define $S = T \cap R_G(P)$, and let $|S| = t$. Consider running lines 5-13 of $\Seq(G,k,r)$ (\cref{algo:seq}) with this choice of $S$. Then there exists a partition of $P$ into exactly $t+1$ (possibly empty) subpaths $P_1, P_2, \cdots, P_{t+1}$ which satisfies the following conditions:
\begin{enumerate}
    \item If $S$ contains a bridge for $P$, then all the $P_j$ are empty. 
    \item If $S$ contains no bridges for $P$, then the vertex disjoint union of the $P_j$ is exactly $P$.
    \item Each $P_j$ is inside some $G[V_{f(j)}]$ generated by $\Seq$ for a recursive execution for some $f(j)$.
\end{enumerate}
Further, 
\[
\mathbb{E}_{|S| = t} \left[ \sum_i s(P_i, G[V_{f(i)}]) \right] \leq \frac{2}{t+1} \cdot s(P,G).
\]
Here, the expectation is conditioned on the event that $|S| = t$; equivalently, we may take the expectation over a uniformly random subset of $t$ elements from $R_G(P)$.
\end{lemma}
Before we prove this lemma, we describe our general proof strategy. Depending on what $S$ is, we will construct a partition of $P$ which satisfies our four constraints. If the set $S$ contains a bridge for $P$, our partition will be empty: $P$'s endpoints are connected through the bridge. If $S$ does not contain any bridges, we simply consider all nonempty subpaths of $P$ inside the recursively generated subproblems induced by $S$: we will show that this partition does not form too many pieces. To prove the expected decrease in the number of path-related nodes, we will explicitly use the randomness of $S$. Consider any set $C$ of ancestors and decendants of $P$. Let $v$ be a vertex inside $C$, and imagine shortcutting $P$ with $C-\{v\}$ and forming the subpaths by our partition. Now consider the event that $v$ is path-relevant for one of these induced subpaths. We will show that there are at most two choices of $v$ from $C$ such that this happens. The result follows with some computation.
\begin{lemma}
    	\label{seq:1llemma}
        Let $P$ be a path, and let $A = \{a_1, a_2, \cdots, a_{l+1}\}$ be ancestors (resp. descendants) of nodes in $P$. If we pick $a_i$ at random from this collection and shortcut using the other $l$ points, $a_i$ is path-relevant for one of the subproblems with probability at most $1/(l+1)$. 
    \end{lemma}
\begin{proof}
         Index the path from head to tail \footnote{We assign the head an index of 0.}, and let $\alpha(x)$ be the lowest-indexed node $p$ such that $x \preceq p$. Assume node $a \in A$ remains path-relevant after shortcutting through all the other $a_i$.  We will show that $a$ must satisfy the following: 
        \begin{itemize}
            \item $a$ is a strictly minimal element amongst the $A$: no $a_i$ can have $a$ as an ancestor.
            \item Amongst all $a_i$ which are unrelated to $a$, $\alpha(a) < \alpha(a_i)$.
        \end{itemize}
        For the first claim, assume for the sake of contradiction that some $a_j$ had $a$ as an ancestor. Note that since we shortcut from $a_j$ it labels $a$ with either $a_j^{\Anc}$ or $\xmark$ depending on whether $a$ an reach $a_i$. In the latter case we are done since we do not recurse on nodes which receive $\xmark$. In the former case we observe that all the $a_i$ can reach the path yet no node in the path can reach any $a_i$ since the $a_i$ are all ancestors. Thus every node $p \in P$ receives either no label from $a_j$ or a $a_j^{\Des}$ label. As $a$ receives a different label from every node in $P$ we conclude that it cannot be path-relevant to any path subproblem in the next level of the recursion.
        
        For the second claim, assume there existed $a_j$ which was unrelated to $a$ such that $\alpha(a_j) \leq \alpha(a)$. This implies that if $a$ is an ancestor to a node $p \in P$ then $a_j$ is also an ancestor for $p$. Thus every node in the path is either unrelated to $a$ or a descendant of $a_j$. As all of $a_j$'s descendants get a label $a_j^{\Des}$ yet $a$ receives no label from $a_j$, we conclude that $a$ can only be placed in a subproblem (if at all) with a piece of the path that it is unrelated to: thus $a$ ceases to be path relevant. 
        
        We now show that there is at most one vertex amongst $A$ which satisfies both of these conditions. Assume for the sake of contradiction that both $a_i$ and $a_j$ would remain path relevant if we shortcutted through $A \bs a_i$ and $A \bs a_j$ respectively. If these two vertices were related, then by the first of our conditions the one which was an ancestor will not remain path-related; contradiction. If the two vertices were unrelated, then $\alpha(a_i) \geq \alpha(a_j)$ or vice-versa: the vertex with the smaller $\alpha$ value cannot remain path related by our second condition. Thus at most one node satisfies our condition, and the claim follows. 

        An identical proof can be used in the case where the $a_i$ are all path descendants. 
    \end{proof}
With this, we complete the proof of \cref{seq:2k+1lemma}.
\begin{proof}[Proof of \cref{seq:2k+1lemma}]
We begin by defining the partition $P_i$. If $S$ contains a bridge for $P$, we set all the $P_i = \emptyset$: this clearly satisfies the conditions. If $S$ does not contain a bridge, we look at the induced subgraphs generated by $\Seq$ in a recursive call. Let $Q_i$ denote $P$ intersected with $V_i$-- the part of $P$ that lies in $G[V_i]$. We choose the $P_j$ to be exactly those $Q_i$ which are nonempty. We need only show that there are at most $|S|+1$ nonempty $P_j$: conditions $1, 2, 3$ are trivial. Observe that as $S$ does not contain any bridges, every path-relevant shortcutter picked is either an ancestor or a descendant of $P$. Thus every shortcutter $s \in S$ induces a ``cut" of $P$ into two contiguous pieces each assigned a different label from $P$. It is straightforward to verify that this implies that $P$ can be split into at most $|S|+1$ contiguous regions each internally with the same labels: this forms our partition $P_i$. 

We now turn our attention to the second fact. Consider picking a random point $v$ from $R_G(P)$. We will show that $v$ is counted in some $s(P_i, G[V_{f(i)}])$ (that is, it remains path-relevant for some subpath) with probability at most $\frac{2}{t+1}$.  Let the points of $S$ be $s_1, s_2, \cdots, s_t$, and observe that the $t+1$ points $v, s_1, \cdots, s_t$ form a random sample from $s(P,G)$. Amongst the $s_1, \text{\ldots} s_t$ we chose, some points are path ancestors, some are path descendants, and some are both (bridges). Assume that there are $\alpha$ ancestors, $\delta$ descendants, and $\beta$ bridges amongst the $s(P,G)$ path related points. We relabel the $s_i$ as $x_1, \cdots, x_a$, $y_1, \cdots, y_d$, and $z_1, \cdots, z_b$ where the $x$ are path ancestors, the $y$ are path descendants, and the $z$ are path bridges. Here, among the $s_i$ there are $a$ ancestors, $d$ descendants, and $b$ bridges. There are two cases: either we picked at least one path bridge (in which case the endpoints of the path are linked and the path gets completely resolved) or we picked $0$ bridge nodes. The probability that we picked $0$ bridge vertices is at most 
    \[
    \left(1 - \frac{\beta}{s(P,G)} \right)^t.
    \]
Now as $v$ is also a randomly chosen vertex: with probability $p = \frac{\alpha}{s(P,G)}$ it is an ancestor, with probability $q = \frac{\delta}{s(P,G)}$ it is a descendant, and with probability $g = \frac{\beta}{s(P, G)}$ it is a bridge. If we condition on the event that none of the $s_i$ were bridges, we observe that $a \sim B(t,\frac{p}{p+q})$ and $d \sim B(t,\frac{p}{p+q})$. Thus the probability $v$ is path relevant is at most 
    \[
    \mathbb{E} \left[ \frac{p}{a+1} + \frac{q}{d+1} \right] + g.
    \]
    by simply applying \cref{seq:1llemma} to the cases where $v$ is an ancestor and $v$ is a descendant separately and union bounding the events. We now recall a useful fact: if $X \sim B(n,p)$, $\mathbb{E}[(X+1)^{-1}] = \frac{1 - (1-p)^{n+1}}{p(n+1)} < \frac{1}{p(n+1)}$. By applying this fact we observe that the probability $v$ is path relevant conditioned on us never picking a bridge shortcut is at most $\frac{2}{t+1} + g$.
    Thus by multiplying by the chance of us never picking a bridge, we see that the final probability $v$ survives is at most 
    \[
        \Big(1 - \frac{\beta}{s(P,G)} \Big)^t \Big( \frac{2}{t+1} + g \Big) = (1-g)^t \Big( \frac{2}{t+1} + g \Big).
    \]
    This can be verified to be at most $\frac{2}{t+1}$, and the result follows. 
\end{proof}

\subsection{Recursion and Inductive Diameter Bound}
\label{sec:recurse_and_diam}

With this helper lemma in place, we now use it to prove our claimed diameter bound. 

\begin{lemma}[Inductive diameter bound]
\label{seq:diambound}
Consider running $\Seq(G, k, 0)$, and consider a recursive execution $\Seq(G[V'], k, r)$ with path $P'$ inside $G[V']$. Let $t \leq \log_k n$ be the largest value of $r$ ever encountered in our recursive calls.  If we complete our algorithm's recursion from $\Seq(G[V'], k, r)$, the expected distance from $\head(P')$ to $\tail(P')$ after applying our computed shortcuts is at most $(4\sqrt{2})^{\br} s(P', G')^{1/2}$, where $\br = \log_k n - r.$ \footnote{We note that the bound obtained here is weaker than the one obtained in the parallel setting. We give a less tight analysis for this lemma in pursuit of a shorter proof.}
\end{lemma}
\begin{proof}
We proceed by induction on $r$ with base case $r = \log_k n$. Note that at this stage we have no more recursion to do: $P'$ must consist of a single node and thus the distance from $\head(P')$ to $\tail(P')$ is $0$. Therefore, the result holds for $r = \log_k n$.

For our induction step, consider an inductive execution of algorithm $\Seq(G[V'], k, r)$, and assume the result for depth $r+1$. Consider the subexecutions directly induced by $\Seq(G[V'], k, r)$. We shortcut $P'$ in the following way. Say that our algorithm chose a set $S$ of $t$ vertices in $R_G(P')$ as shortcutters. If one of the vertices in $S$ is a bridge of $P'$, we would simply traverse the bridge and go from our path's head to tail in $2$ edges. Otherwise, by \cref{seq:2k+1lemma} we would split $P'$ into $t+1$ subpaths $P'_1, P'_2, \cdots, P'_{t+1}$ where
\begin{itemize}
    \item The disjoint vertex union of the $P'_i$ is $P'$.
    \item Each $P'_i$ is inside some $G[V_{f(i)}]$ on which $\Seq$ executes $\Seq(G[V_{f(i)}], k, r+1)$ on for some $f(i)$.
\end{itemize}

We can get from the head to tail of $P'$ by inductively traversing each $P'_i$ in the order we encounter them, while using $t$ extra edges to go between these paths. Now in either case we use at most $t+2$ edges to traverse between the subpaths formed by our recursion. We additionally use some number of edges to traverse from the tail to head of each of the subpaths we form. By our inductive hypothesis, we see that the length of each $P'_i$ after applying our shortcuts satisfies
\[
\mathbb{E}\left[\text{shortcut length of } P'_i \right] \le (4 \sqrt{2})^{\br-1} \mathbb{E} [s(P'_i, G[V'_i])^{1/2}]. 
\]
Thus if we condition on the fact that $|S| = t$ the expected shortcut length of $P'$ is at most 
\begin{align*}
\mathbb{E}\left[\text{shortcut length of } P' \Big| |S| = t \right] &\leq t+2 +  \sum_{i=1}^{t+1} (4 \sqrt{2})^{\br-1} \mathbb{E} \left[s(P'_i, G[V'_i])^{1/2} \Big| |S| = t \right] \\
&\leq 2 + \sum_{i=1}^{t+1} (2 \sqrt{2})^{-1} (4 \sqrt{2})^{\br} \mathbb{E} \left[s(P'_i, G[V'_i])^{1/2} \Big| |S| = t \right]
\end{align*}
as $(4 \sqrt{2})^{\br-1} \mathbb{E} [s(P'_i, G[V'_i])^{1/2}] \geq 1$ for every $i$. We obtain 
\begin{align*}
\mathbb{E}\left[\text{shortcut length of } P'  \Big| |S| = t \right] &\leq 2 +  \sum_{i=1}^{t+1} (2 \sqrt{2})^{-1} (4 \sqrt{2})^{\br} \mathbb{E} \left[s(P'_i, G[V'_i])^{1/2} \Big| |S| = t \right] \\
&\leq 2 +  \sum_{i=1}^{t+1} (2 \sqrt{2})^{-1} (4 \sqrt{2})^{\br} \mathbb{E} \left[s(P'_i, G[V'_i]) \Big| |S| = t \right]^{1/2} \\
&\leq 2 + (2 \sqrt{2})^{-1} (4 \sqrt{2})^{\br} \sqrt{t+1} \mathbb{E} \left[ \sum_{i=1}^{t+1} s(P'_i, G[V'_i]) \Big| |S| = t \right]^{1/2} \\
&\leq 2 + (2 \sqrt{2})^{-1} (4 \sqrt{2})^{\br} \sqrt{t+1} \left[\frac{2}{t+1} s(P', G[V']) \right]^{1/2} \\
&\leq  2 + 2^{-1} (4 \sqrt{2})^{\br} s(P', G[V'])^{1/2} \\
&\leq (4 \sqrt{2})^{\br} s(P', G[V'])^{1/2},
\end{align*}
where we used Jensen's inequality in the second inequality, Cauchy-Schwarz in the third, and \cref{seq:2k+1lemma} in the fourth.  As our final bound is independent of the value of $t$, the result follows.
\end{proof}

With this result in place, we can now prove our final theorem statement by setting $r = 0$ in \cref{seq:diambound}.

\begin{theorem}
\label{thm:main5}
Let $G = (V,E)$ be a digraph with $n$ nodes and $m$ edges. Then with probability $1 - 2 n^{-10}$ algorithm $\Seq(G, k, 0)$ runs in $\otilde(mk)$ time and constructs a set $F$ of $\otilde(nk)$ shortcuts such that an arbitrary path $P$ is shortcut to length $n^{1/2 + O(1/\log k)}$ w.h.p.
\end{theorem}
\begin{proof}
By \cref{seq:workbound} our claimed work and shortcut bound follow immediately. Let $P$ be any path from $s$ to $t$. By \cref{seq:diambound} we observe that the expected length of $P$ after applying our constructed shortcuts is at most $(4 \sqrt{2})^{\log_k n} (s(P,G))^{1/2}$. As $s(P,G) \leq n$, this equals $n^{1/2 + O(1/\log k)}$ as claimed.
\end{proof}
We observe that \cref{thm:seq} follows from this by calling $\Seq(G,k,0)$ $O(\log n)$ times: the probability that any given pair $s,t$ fails to have a path of length twice the bound in \cref{thm:main5} between them after $c \log n$ calls is at most $1/n^c$ by Markov's inequality. By union bounding over all $n^2$ pairs of points the claim follows.

\section{Parallel Algorithm}
\label{sec:parallel}

In this section, we show how to extend the nearly linear work sequential shortuctting algorithm of \cref{sec:seq} into a work-efficient, low depth parallel algorithm. In \cref{sec:parallelprelim} we extend the definitions in \cref{sec:prelim} to the setting of distance limited searches and reachability, which we need throughout the section. In \cref{sec:parallelalgo} we state our main algorithms \ParallelSC~and \ParallelDiam, and give intuition for how to reason about them. In \cref{sec:parallelwork} we bound the total work, depth, and number of shortcut edges added in the algorithm. Finally, in \cref{sec:paralleldiam} we bound the diameter of the resulting graph after the execution of our algorithms.

\subsection{Notation for Distance Limited Searches}
\label{sec:parallelprelim}
\paragraph{Digraph distances and distance-limited relations:} Let $G$ be a digraph. For vertices $u, v \in V$ define $d(u, v)$ to be the length of the shortest path from $u$ to $v$. We define $d(u, v) = +\infty$ if $v \not\in \RD_G(u)$, i.e. $v$ is not reachable from $u$. We also define distance-limited reachability, which is a natural extension of the notion of reachability defined in \cref{sec:prelim}. For a parameter $D$, we define the $D$-descendants, $D$-ancestors, and $D$-related vertices to $v$ as \[ \RD_G(v, D) = \{ u \in V : d(v, u) \le D \} \text{ and } \RA_G(v, D) = \{ u \in V : d(u, v) \le D \} \] \[ \text{ and } R_G(v, D) = \RD_G(v, D) \cup \RA_G(v, D). \] We extend all this notation to subgraphs $G'$ of $G$ in the natural way. Define $d_{G'}(u, v)$ to the length of the shortest path from $u \to v$ only using only vertices and edges in $G'$. Then we define $\RD_{G'}(v, D)$, $\RA_{G'}(v, D)$, and $R_{G'}(v, D)$ as above. A vertex $u$ is \emph{unrelated within distance $D$} to a vertex $v$ (with respect to a subgraph $G'$) if $u \in V[G'] \bs R_{G'}(v, D).$

\paragraph{Distance limited path relations.} In Algorithm \ParallelSC, our searches are limited to distance $\kappa D$ for some random parameter $\kappa.$ Here we define distance limited path relations, analogous to the vertex and path relations defined in \cref{sec:prelim}.
As before, we denote a path $P = \l v_0, v_1, \dots, v_\ell \r$, where all the $v_i$ are vertices of $G$. We now make the following definitions. As the range the parameter $\kappa$ is chosen from depends on $r$, the current recursion depth of the algorithm, our below definitions also depend on $r$. This dependence is made explicit in lines \ref{line:k2r} and \ref{line:pickk} of Algorithm \ParallelSC.

%\ynote{I'll decide how to precisely define these later. There's a question of whether you want to just define wrt head and tail or wrt the whole path.}
Throughout, we say that $u \pe^s v$ if $d_G(u, v) \le s$, i.e. there is a path of length at most $s$ from $u$ to $v$. We say that $u \not\pe^s v$ if $d_G(u, v) > s.$
\begin{itemize}
\item \textbf{Fully path-related vertices.} We say that vertex $v$ is fully path-related if $v \in R_G(P, \kappa_{2r+1}D).$ In other words, $v \in R_G(P, \kappa D)$ for all $\kappa \in [\kappa_{2r+1}, \kappa_{2r}].$ 

\item \textbf{Partially path-related vertices.} We say that a vertex $v$ is partially path-related if $v \in R_G(P, \kappa_{2r}D).$ In other words, $v \in R_G(P, \kappa D)$ for some $\kappa \in [\kappa_{2r+1}, \kappa_{2r}].$

We distinguish three types of fully or partially path-related vertices. The below definitions depend on the parameter $\kappa$ chosen.

\item \textbf{$\kd$-Descendants.} We say that a vertex $v$ is a $\kd$-descendant of the path $P$ if $v \in \RD_G(P, \kd)\bs \RA_G(P, \kd).$ Note that then we have that $v_0 \pe^{(\kappa+1)D} v$ and $v \not\pe^{\kd} v_\ell.$
\item \textbf{$\kd$-Ancestors.} We say that a vertex $v$ is an $\kd$-ancestor of the path $P$ if it is the case that $v \in \RA_G(P, \kd)\bs \RD_G(P, \kd).$ Note that then we have that $v_0 \not\pe^{\kd} v$ and $v \pe^{(\kappa+1)D} v_\ell.$
\item \textbf{$\kd$-Bridges.} We say that a vertex $v$ is a $\kd$-bridge of the path $P$ if $v \in \RD_G(P, \kd)\cap \RA_G(P, \kd).$ Note that then we have that $v_0 \pe^{(\kappa+1)D} v$ and $v \pe^{(\kappa+1)D} v_\ell.$
\end{itemize}

Now, we define the following quantities and briefly explain their importance in \cref{algo:parallelsc}. Further details are explained in the paragraph below (explanation of \cref{algo:parallelsc} and \cref{algo:paralleldiam}). Let $G$ be the $n$-vertex $m$-edge digraph which we input, with vertex set $V$ and edge set $E$.
\begin{itemize}
\item \textbf{Inputs $k, r, r^\fringe$:} The input $k$ denotes the speed that our algorithm recurses at, and intuitively digraphs one level lower in the recursion are ``smaller" by a factor of $k$. This is made precise by \cref{par:pararvbound}. The level $r \le \log_k n$ denotes the level of recursion the algorithm is at. Additionally, we have an inner recursion level $r^\fringe \le \log n$ for the ``fringe vertices" (defined below). For each level of recursion $r$, it has at most $\log n$ inner levels of recursion.
\item \textbf{Set $S$:} $S$ is the set of vertices from which we search and build shortcuts from.
\item \textbf{Set $F$:} $F$ is the final set of shortcuts we construct.
\item \textbf{Search scale / distance $D$:} The parameter $D$ denotes the scale on which the algorithm is runs breadth first searches. Our main claim is that through one call of $\ParallelSC(G, k, 0, 0)$ an arbitrary path in $G$ of length $D$ will be shortcut to expected length $\frac{D}{10}$ (\cref{par:paradiambound}).
\item \textbf{Parameters $\kappa_i, \kappa$:} In the $r$-th recursion level, we randomly choose the parameter $\kappa \in [\kappa_{2r+1}, \kappa_{2r}]$ to obtain our search radius $\kappa D$ for our breadth first searches.
\item \textbf{Probability $p_r$:} At recursion depth $r$, for each vertex $v \in V(G)$, we put $v$ in $S$ with probability $p_r$.
\item \textbf{Labels $\vD, \vA, \xmark$:} We want to distinguish vertices by their relations to vertices in  the set $S$. Therefore, when we search from a vertex $v$ we assign a label $\vD$ to add vertices in $\RD(v, D)\bs \RA(v, D)$, label $\vA$ to all vertices in $\RA(v, D)\bs \RD(v, D)$, and label $\xmark$ to all vertices in $\RD(v, D) \cap \RA(v, D).$ The label $\xmark$ should be understood as ``eliminating" the vertex (since it is in the same strongly connected component as $v$ and we have shortcut through $v$ already).
\end{itemize}

\subsection{Main Algorithm Description}
\label{sec:parallelalgo}
\begin{algorithm}[p]
\caption{$\ParallelSC(G, k, r, r^\fringe)$. Takes a digraph $G$, parameter $k$, recursion depth $r \le \log_k n$ (starts at $r = 0$), and inner fringe node recursion depth $r^\fringe \le \log n$. Returns a set of shortcut edges to add to $G$. $n$ denotes the number of vertices at the top level of recursion, not the number of vertices in $G$.}
\begin{algorithmic}[1]
\State $p_r \assign \min\left(1, \frac{10 k^{r+1} \log n}{n} \right).$ \label{line:setpr}
\State $S \assign \emptyset.$ \label{line:sets}
\State $\kappa_{2r+1} \assign 10^6 k^2 \log^5 n \left(1 + \frac{1}{4\log n} \right)^{-2r-1}$ and $\kappa_{2r} \assign 10^6 k^2 \log^5 n \left(1 + \frac{1}{4\log n} \right)^{-2r}.$ \label{line:k2r}
\State Choose $\kappa \in [\kappa_{2r+1}, \kappa_{2r}]$ uniformly at random. \Comment{Picking a random search distance} \label{line:pickk}
\State $D \assign 100 \cdot \sqrt{2}^{\log_k n} \cdot n^\frac12 \cdot \log^2 n.$ \label{line:setd}
\For{$v \in V$} \label{line:makes1}
	\State With probability $p_r$ do $S \assign S \cup \{v\}.$ \label{line:makes2}
\EndFor
\State $F \assign \emptyset.$ \label{line:setf}
\For{$v \in S$} \label{line:bfsstart}
	\For{$w \in \RD(v, (\kappa+1)D)$} add edge $(v, w)$ to $F$. \label{line:addsc1}
	\EndFor
	\For{$w \in \RA(v, (\kappa+1)D)$} add add $(w, v)$ to $F$.	\label{line:addsc2}
	\EndFor
	\For{$w \in \RD(v, \kd) \bs \RA(v, \kd)$} add label $\vA$ to vertex $w$. \label{line:addanc}
	\EndFor
	\For{$w \in \RA(v, \kd) \bs \RD(v, \kd)$} add label $\vD$ to vertex $w$. \label{line:adddes}
	\EndFor
	\For{$w \in \RD(v, \kd) \cap \RA(v, \kd)$} add label $\xmark$ to vertex $w$. \label{line:addxmark}
	\EndFor
	\State $V^\ring_v \assign R(v, (\kappa+1)D) \bs R(v, (\kappa-1)D).$ \label{line:bfsend}
	\State $F \assign F \cup \ParallelSC(G[V^\ring_v], k, r, r^\fringe+1).$ \Comment{Recursion on fringe nodes} \label{line:fringerecurse}
\EndFor
\State $W \assign \{ v \in V : v \text{ has no label of } \xmark \}.$ \label{line:remxmark}
\State $V_1, V_2, \dots, V_\ell \assign$ partition of $W$ such that $x, y \in V_i$ iff $x$ and $y$ have exactly the same labels. \Comment{Vertices in the $V_i$ have no label of $\xmark$.} \label{line:vi}
\For{$1 \le i \le \ell$} \label{line:normalrecurse0}
	\State $F \assign F \cup \ParallelSC(G[V_i], k, r+1, 0).$ \label{line:normalrecurse}
\EndFor
\State \Return $F$
\end{algorithmic}
\label{algo:parallelsc}
\end{algorithm}

\begin{algorithm}[p]
\caption{$\ParallelDiam(G, k)$. Takes a digraph $G = (V, E)$, parameter $k$. Modifies digraph $G$. Parallelizable diameter reduction algorithm.}
\begin{algorithmic}[1]
\For{$i = 1$ to $10\log n$}
	\For{$j = 1$ to $10\log n$}
		\State $S_j \assign \ParallelSC(G, k, 0, 0)$, aborting if the work or shortcut edge count exceeds $10$ times the bound in \cref{par:finalparabound}.
	\EndFor
	\State $E(G) \assign E(G) \cup \left(\bigcup_j S_j \right)$
\EndFor
\end{algorithmic}
\label{algo:paralleldiam}
\end{algorithm}

\paragraph{Overview of \cref{algo:parallelsc}.} \cref{algo:parallelsc} is similar to \cref{algo:seq}. All parts of \cref{algo:seq} can be implemented in low parallel depth except for the breadth first searches from the vertices in $S$ (lines \ref{line:bfsstart} to \ref{line:bfsend}). To resolve this, a natural idea is to limit the distance of the breadth first searches to $D$, where $D$ denotes the diameter bound given by \cref{seq:diambound}. Running these incomplete searches introduces issues in the analysis though. To get around this, we follow an approach similar to \cite{Fine18} and perform some additional computation on the fringe of our breadth first searches. Specifically, we choose a random integer $\kappa$ in the range $[\kappa_{2r+1}, \kappa_{2r}]$ (think of these as parameters which are $\poly(\log n, k)$), and search from a vertex $v$ to distance approximately $\kappa D.$ Then, we call the vertices in the set $R(v, (\kappa+1)D) \bs R(v, (\kappa-1)D)$ the \emph{fringe vertices}. We chose $\kappa$ randomly to ensure that the expected number of fringe vertices is sufficiently small. We then recurse on the fringe vertices, which is done in line \ref{line:fringerecurse} of \cref{algo:parallelsc}.

In addition, we assign labels based on reachability \emph{within distance $\kappa D$}. As we show later in the section, the analysis as done in \cref{sec:seq} can be modified to tolerate these changes and obtain a similar result.

We would like to note some differences between the ways fringe vertices are handled in our algorithms compared to those in \cite{Fine18}. One difference is that we directly handle fringe vertices by recursing on only that set, while in Fineman's algorithm the fringe vertices are lumped into the recursion on ancestor sets. The reason for the difference is that our way of partitioning our vertex set before recursing is more involved. Additionally, in our algorithm, we explicitly track the depth of recursion on fringe vertices (parameter $r^\fringe$) inside our algorithm. We must do this as our algorithm requires good control on the number of ancestors and descendants of a vertex (\cref{seq:rvbound}), and we do not obtain the required bound on the number of ancestors or descendants when we recurse on fringe vertices. Additionally, we must ensure that $r^\fringe \le \log n$ so that our algorithm still has low parallel depth.

\paragraph{Explanation of \cref{algo:parallelsc} and \cref{algo:paralleldiam}.} Here, we give a detailed description of what each part of \cref{algo:parallelsc} and \cref{algo:paralleldiam} is doing. Lines \ref{line:setpr} and \ref{line:sets} of \cref{algo:parallelsc} are simply initializing the set of vertices we search from $S$ to the empty set and picking the probability $p_r$ for which $v \in S$. In line \ref{line:setd}, we choose our search scale $D$. This is chosen to be a constant factor larger than the bound given in \cref{seq:diambound} and the analogous \cref{par:paradiambound}. In lines \ref{line:k2r} and \ref{line:pickk}, we define the parameters $\kappa \in [\kappa_{2r+1}, \kappa_{2r}]$, which will define our search radius $\kappa D.$ Note that $\kappa_0 \ge \kappa_1 \ge \kappa_2 \ge \cdots$ and that $\kappa_i \ge \frac{1}{2}\kappa_0$ for all $i$, as $r \le \log_k n.$ This way, the search radius is decreasing with every recursion level. In lines \ref{line:makes1} and \ref{line:makes2}, we are choosing the set of vertices to breadth first search from, and we are adding them to $S$. In line \ref{line:setf}, we initialize the set of shortcut edges we are going to eventually add to the empty set.

In lines \ref{line:addsc1} and \ref{line:addsc2}, we are adding all the shortcut edges through a vertex $v \in S$. In lines \ref{line:addanc} to \ref{line:addxmark} we are applying the ancestor, descendant, and ``eliminated" labels to other vertices. In line \ref{line:fringerecurse}, we are running a recursion on the fringe vertices in our search from $v$. Note that we have increased $r^\fringe$ to $r^\fringe+1$ but have kept the parameter $r$ the same.

In lines \ref{line:remxmark} and \ref{line:vi}, we are processing the labels assigned to the vertices. We first remove from consideration all vertices which have a $\xmark$ label, and then we partition the remaining ones into groups based on matching labels. In lines \ref{line:normalrecurse0} and \ref{line:normalrecurse} we recurse on these groups we have created. Note that we have increased $r$ to $r+1$ and reset $r^\fringe$ to $0$.

In \cref{algo:paralleldiam} we are simply running \ParallelSC~for multiple iterations. Specifically, we can guarantee that the expected head to tail distance of any path of original length at most $D$ is now at most $\frac{D}{10}$. Running this $O(\log n)$ times ensures that any fixed path of length $D$ has been shortcut to length $\frac{D}{5}$ with high probability. Now, if we run this procedure $O(\log n)$ times, it is easy to see that any path's head to tail distance will get reduced to $D$ with high probability: split this path into polynomially many paths of length at most $D$ and note that with high probability each of these paths' lengths gets reduced by a constant factor. Thus the original path's length falls by a constant factor, and we continue to do this until the original path's length becomes at most $D$. 

\paragraph{How to reason about the randomness.} In order to reason about the randomness used in sampling $\kappa$ and $S$, we should imagine that the during an execution of $\ParallelSC(G', k, r, r^\fringe)$ the algorithm first samples $\kappa$ before doing anything else. After sampling $\kappa$, we then know precisely which vertices are $\kd$-descendants, $\kd$-ancestors, and $\kd$-bridges with respect to the specific path $P$ that we are analyzing. Given this, we can essentially proceed forwards with similar arguments to those in \cref{sec:seq}.

\subsection{Shortcut and Work Bound}
\label{sec:parallelwork}
In this section, we prove many lemmas which help us bound the total amount of work and shortcut edges added. We attempt to make the bounds with high probability whenever possible, but some are in expectation.

We start by proving an analogue of \cref{seq:rvbound}.
\begin{lemma}
\label{par:pararvbound}
Consider an execution of $\ParallelSC(G, k, 0, 0)$ on $n$-node $m$-edge digraph $G$. With probability $1 - n^{-4}$ in each recursive execution of the form $\ParallelSC(G',k,r,r^\fringe)$ of Algorithm~\ref{algo:parallelsc} the following holds:
\[
|\RD_{G'}(v, \kappa_{2r}D)| \leq n k ^{-r} ~~~ \text{and} ~~~ |\RA_{G'}(v, \kappa_{2r}D)| \leq n k^{-r}
~~~ \forall v \in V(G') ~.
\]
\end{lemma}
\begin{proof}
We proceed by induction. The base case $r = r^\fringe = 0$ is clear. All recursive calls are one of the following forms: let $G'$ be a digraph in which we made a call to $\ParallelSC(G', k, r, r^\fringe)$. We consider the cases of a recursive call to $\ParallelSC(G'[V^\ring_v],k,r,r^\fringe+1)$ and a recursive call to $\ParallelSC(G'[V_i],k,r+1,0)$ separately (line \ref{line:fringerecurse} and \ref{line:normalrecurse}). In the former case, the claim is clear, as in our recursive call to $\ParallelSC(G'[V^\ring_v],k,r, r^\fringe+1)$ the parameter $r$ stays fixed, and the digraph $G'$ already satisfied \[
\RD_{G'}(v, \kappa_{2r}D) \leq n k ^{-r} ~~~ \text{and} ~~~ \RA_{G'}(v, \kappa_{2r}D) \leq n k^{-r}
~~~ \forall v \in V(G') ~.
\]
Now we consider the case referring to $G'[V_i]$. Consider a vertex $v \in V_i$, and let $\kappa$ be the parameter chosen by the algorithm. Let the vertices in $\RD_{G'}(v, \kd)$ be $x_1, x_2, \dots, x_M$, where $M = |\RD_{G'}(v, \kd)|.$ For $1 \le i \le M$, define \[ S_i = \{ x_j : 1 \le j \le M, x_i \in \RD_{G'}(x_j,\kd) \text{ or } x_i \text { unrelated within distance } \kd \text{ to } x_j \}. \] In other words, $S_i$ consists of all vertices $x_j$ such that $x_i$ is a $\kd$-descendant of $x_j$ or $x_i$ is unrelated within distance $\kd$ to $x_j$.

We now show that if $x_j$ is in the shortcutter set $S$ and $x_j \in S_i$, then $x_i \notin \RD_{G'[V_i]}(v,\kd)$, i.e. is not a $\kd$-descendant of $v$ in the recursive subproblem induced. Indeed, $v$ would receive a $x_j^\Anc$ label, while $x_i$ receives either a $x_j^\Des$ label, unrelated label, or $\xmark$. Because $\kappa_{2r+2} \le \kappa$, then $x_i \notin \RD_{G'[V_i]}(v,\kappa_{2r+2}D)$ also. Now, if $|S_i| \ge nk^{-r-1}/2$, then the probability that no $x_j \in S_i$ is chosen to be in the shortcutter set $S$ is at most $(1-p_r)^{|S_i|} \le n^{-5}$ by our choice of $p_r$. A union bound over all $i$ gives that the failure probability is bounded by $n^{-4}.$

To finish the proof, we argue that the number of indices $i$ with $|S_i| < nk^{-r-1}/2$ is at most $nk^{-r-1}.$ Define \[ Y = \{ 1 \le i \le M : |S_i| < nk^{-r-1}/2 \}. \] Note that for all $1 \le i,j \le M$, either $x_i \in S_j$ or $x_j \in S_i.$ Therefore, some $i \in Y$ satisfies $|S_i| \ge \frac{|Y|}{2}$, as $x_i \in S_i$ also, so $|Y| \le 2|S_i| \le nk^{-r-1}$ as desired.
\end{proof}
As we must recurse separately on fringe vertices (hence redoing the computation on them), we need to be able to control the number of fringe vertices. Therefore, it is natural to  bound the expected number of times a fixed vertex is a fringe vertex for some breadth-first search.
\begin{lemma}
\label{par:singlefringebound}
Consider an execution of $\ParallelSC(G, k, 0, 0)$ on $n$-node $m$-edge digraph $G$ and a recursive execution of the form $\ParallelSC(G',k,r,r^\fringe)$. Let $u \in V(G').$ The expected number of times $u$ is in a recursive fringe subproblem, i.e. \[ u \in R_{G'}(v, (\kappa+1)D))\bs R_{G'}(v, (\kappa-1)D)) \] for some $v \in S$ (line \ref{line:fringerecurse}), is at most $\frac{1}{1000 k \log^3 n}.$
\end{lemma}
\begin{proof}
For a vertex $v \in R_{G'}(u, \kappa_{2r}D)$, the probability that $v \in S$ and $u \in R_{G'}(v, (\kappa+1)D))\bs R_{G'}(v, (\kappa-1)D))$ over a random $\kappa \in [\kappa_{2r+1}, \kappa_{2r}]$ is clearly at most $p_r \cdot \frac{2}{\kappa_{2r} - \kappa_{2r+1}}.$ Also, we have that $|R_{G'}(u, \kappa_{2r}D)| \le \frac{2n}{k^r}$ with high probability by \cref{par:pararvbound}, so the expected number of times $u$ is in a fringe subproblem is at most \[ p_r \cdot \frac{2}{\kappa_{2r} - \kappa_{2r+1}} \cdot \frac{2n}{k^r} \le \frac{1}{1000 k \log^3 n} \] from our choice of $\kappa_{2r}, \kappa_{2r+1}, p_r$.
\end{proof}
Now, we can show that the total size of all the digraphs we process during the algorithm doesn't increase much between levels of recursion. Additionally, at deep fringe recursion levels (for large $r^\fringe$) exponentially few vertices are processed.
\begin{lemma}[Expected total size bounds]
\label{par:expsize}
Consider an execution of $\ParallelSC(G, k, 0, 0)$ on $n$-node $m$-edge digraph $G$. For any recursion depth $r \le \log_k n$ and fringe recursion depth $r^\fringe$ we have that:
\begin{itemize}
\item The expected value of the total number of vertices of the digraphs $G'$ in all recursive executions $\ParallelSC(G',k,r,r^\fringe)$ is at most $n \cdot \left(1 + \frac{1}{\log n}\right)^r \cdot \frac{1}{(2 \log n)^{r^\fringe}}.$
\item The expected value of the total number of edges of the digraphs $G'$ in all recursive executions $\ParallelSC(G',k,r,r^\fringe)$ is at most $m \cdot \left(1 + \frac{1}{\log n}\right)^r \cdot \frac{1}{(2 \log n)^{r^\fringe}}.$
\end{itemize}
\end{lemma}
\begin{proof}
We focus on proving the first point, as the second is analogous. The first point follows directly by induction from a combination of \cref{par:pararvbound} and \cref{par:singlefringebound}. Specifically, we separate the cases $r^\fringe > 0$ and $r^\fringe = 0.$ In the former case, note that this must result as a recursive call of the form $\ParallelSC(G''[V^\ring_v],k,r,r^\fringe)$, where we also made a recursive call of the form $\ParallelSC(G'', k,r,r^\fringe-1).$ By induction and \cref{par:singlefringebound}, we know that the expected total number of vertices in all recursive calls to $\ParallelSC(G''[V^\ring_v],k,r,r^\fringe)$ is at most \[ \frac{1}{1000 k \log^3 n} \cdot n \cdot \left(1 + \frac{1}{\log n}\right)^r \cdot \frac{1}{(2 \log n)^{r^\fringe-1}} \le n \cdot \left(1 + \frac{1}{\log n}\right)^r \cdot \frac{1}{(2 \log n)^{r^\fringe}} \] as desired. We would like to note that the value $\frac{1}{1000 k \log^3 n}$ is much smaller than we need to prove \cref{par:expsize}; instead it is needed in \cref{par:parapathsplit} below.

In the case $r^\fringe = 0$, it must have resulted in a recursive call to $\ParallelSC(G''[V_i],k,r+1,0)$ for some digraph $G''$, where we also made a recursive call of the form $\ParallelSC(G'', k,r,r^\fringe).$ As the $V_i$ form a partition of $V(G'')$, we can see that the total number of vertices over all calls to $\ParallelSC(G''[V_i],k,r+1,0)$ for such $G''$ and $V_i$ is at most
\[ \sum_{r^\fringe \ge 0} n \cdot \left(1 + \frac{1}{\log n}\right)^r \cdot \frac{1}{(2 \log n)^{r^\fringe}} \le n \cdot \left(1 + \frac{1}{\log n}\right)^{r+1} \] as desired.
\end{proof}
Because exponentially few vertices are processed at large fringe recursion depths, we can guarantee that with high probability that $r^\fringe \le \log n$ always.
\begin{corollary}[Fringe depth is at most logarithmic]
\label{cor:fringedepth}
Consider an execution of $\ParallelSC(G, k, 0, 0)$ on an $n$-node $m$-edge digraph $G$. With probability at least $1 - n^{-10}$ we make no recursive calls of the form $\ParallelSC(G',k,r,r^\fringe)$ and $r^\fringe > \log n$.
\end{corollary}
\begin{proof}
The expected value of the total number of vertices of the digraphs $G'$ in all calls to \\ $\ParallelSC(G', k, r, r^\fringe)$ where $r^\fringe > \log n$ is at most \[ \sum_{r = 0}^{\log_k n} \sum_{r^\fringe > \log n} n \cdot \left(1 + \frac{1}{\log n}\right)^r \cdot \frac{1}{(2 \log n)^{r^\fringe}} \le \frac{3n}{(\log n)^{\log n}} \le \frac{1}{n^{10}} \] by \cref{par:expsize}. Thus, the claim follows.
\end{proof}
Additionally, we can bound the total sizes of all digraphs we process during the algorithm.
\begin{corollary}[Total size bound]
\label{cor:totsize}
Consider an execution of $\ParallelSC(G, k, 0, 0)$ on an $n$-node $m$-edge digraph $G$.
\begin{itemize}
\item The expected value of the total number of vertices in the digraphs $G'$ in all recursive calls to $\ParallelSC(G',k,r,r^\fringe)$ for some $r \le \log_k n, r^\fringe$ is at most $3n \log n$.
\item The expected value of the total number of edges in the digraphs $G'$ in all recursive calls to $\ParallelSC(G',k,r,r^\fringe)$ for some $r \le \log_k n, r^\fringe$ is at most $3m \log n$.
\end{itemize}
\end{corollary}
\begin{proof}
Note that the total number of vertices over all the $G'$ in recursive calls is at most \[ \sum_{r = 0}^{\log_k n} \sum_{r^\fringe \ge 0} n \cdot \left(1 + \frac{1}{\log n}\right)^r \cdot \frac{1}{(2 \log n)^{r^\fringe}} \le 3n \log n \] by \cref{par:expsize}. The edge bound follows similarly.
\end{proof}
We now proceed towards proving our ultimate bounds on expected total work and expected number of shortcut edges added.
\begin{lemma}
\label{par:paralabelbound}
Consider an execution of $\ParallelSC(G, k, 0, 0)$ on an $n$-node, $m$-edge digraph $G$. Consider a recursive execution of $\ParallelSC(G',k,r,r^\fringe)$. For all integers $\kappa \in [\kappa_{2r+1}, \kappa_{2r}]$, with probability $1-n^{-10}$ over the choice of $S$ we have that for all vertices $v \in V(G')$, the number of vertices $u \in S$ such that $v \in R_{G'}(u, \kd)$ is at most $50k\log n.$
\end{lemma}
\begin{proof}
Fix $v \in V(G').$ By \cref{par:pararvbound}, the number of $u$ for which $v \in R_{G'}(u, \kd)$ is at most $2nk^{-r}.$ Therefore, the expected number of $u \in S$ for which $v \in R_{G'}(u, \kd)$ is at most $p_r \cdot 2nk^{-r} \le 20k\log n.$ By a Chernoff bound (\cref{lemma:chernoff}), we have that the number of vertices $u \in S$ such that $v \in R_{G'}(u, \kd)$ is at most $50k \log n$ with probability at least
\[ 1-\exp\left(-11k\log n\right) \ge 1-n^{-11}. \]
The claim follows by union bound.
\end{proof}
\begin{lemma}[Work and depth bound]
\label{par:finalparabound}
An execution of $\ParallelDiam(G, k)$ on an $n$-node, $m$-edge digraph $G$ can be implemented to do $\O(mk+nk^2)$ total work in expectation, add $\O(nk)$ shortcut edges in expectation, and have parallel depth $\O(\poly(k) \sqrt{2}^{\log_k n} n^\frac12)$ with high probability.
\end{lemma}
\begin{proof}
It suffices to show that an execution of $\ParallelSC(G, k, 0, 0)$ on an $n$-node, $m$-edge digraph $G$ can be implemented to do $\O(mk)$ total work in expectation, add $\O(nk)$ shortcut edges in expectation, and have parallel depth $\O(\poly(k) \sqrt{2}^{\log_k n} n^\frac12)$ with high probability. We show this in the two paragraphs below. Then the lemma follows as all digraphs $G$ that we call $\ParallelSC(G, k, 0, 0)$ on during an execution of $\ParallelDiam(G, k)$ will have $\O(m+nk)$ edges with high probability.

\paragraph{Work bound.} By \cref{par:paralabelbound}, it is clear that the number of shortcut edges added is within a multiplicative $\O(k)$ of the total number of vertices in all digraphs in all recursive subproblems with high probability, hence is $\O(nk)$ in expectation by \cref{par:expsize}.

Similarly, the total work from the breadth first searches is within a multiplicative $\O(k)$ of the total number of edges in all digraphs in all recursive subproblems, hence is $\O(mk)$ in expectation by \cref{par:expsize}. The remaining nontrivial work comes from line \ref{line:vi}. We implement line \ref{line:vi} in the following manner. Consider a recursive execution of the algorithm on a digraph $\ParallelSC(G', k, r, r^\fringe).$ We put an arbitrarily total ordering on the labels $v^\Anc$ and $v^\Des$, and for each vertex $u \in V(G')$ we assume that its set of labels, none of which are $\xmark$, are sorted according to this total ordering. Now, we sort all vertices $u \in V(G')$ that have no label which is $\xmark$ lexicographically by their set of labels using Cole's mergesort algorithm \cite{Cole93}, and then group them into the sets $V_i$ based on contiguous groups that have the same label. Each comparison takes work $\O(k)$ with high probability as every vertex $u$ has $\O(k)$ labels with high probability by \cref{par:paralabelbound}. The total work of the merge sort is thus $\O(k |V(G')|)$ as desired.

\paragraph{Depth bound.} Clearly, the breadth first searches (lines \ref{line:bfsstart} to \ref{line:bfsend}) can be implemented in $\O(\kappa_0 D) = \O(\poly(k) \sqrt{2}^{\log_k n} n^\frac12)$ parallel depth by our choice of $\kappa_0.$ Line \ref{line:vi} can be implemented in parallel depth $\O(k)$ via Cole's merge sort \cite{Cole93} as described in the above paragraph.

\end{proof}
\subsection{Diameter Bound after Shortcutting}
\label{sec:paralleldiam}
In this section we show the analogue of \cref{seq:diambound}.
\begin{lemma}[Inductive parallel diameter bound]
\label{par:paradiambound}
Perform an execution of $\ParallelSC(G, k, 0, 0)$ on an $n$-node, $m$-edge digraph $G$. Consider a recursive execution of $\ParallelSC(G',k,r,r^\fringe)$, with all shortcut edges added. For any path $P \subseteq G'$ with length at most $D$, the expected distance from $\head(P)$ to $\tail(P)$ using edges in $G'$ and shortcut edges is at most \[ 5(\br+1) \left( \sqrt{2} + \frac{1}{2\log n}\right)^{\br} s(P, G')^\frac12, \] where $\br = \log_k n - r$, and $s(P, G')$ is the number of partially path related vertices to $P$ in $G'$, i.e. $s(P, G') = |R_{G'}(P, \kappa_{2r}D)|$.
\end{lemma}
Note that we have that \[ 5(\log_k n+1) \left( \sqrt{2} + \frac{1}{2\log n}\right)^{\log_k n} n^\frac12 \le 10\log n \sqrt{2}^{\log_k n} n^\frac12 \le \frac{D}{10}, \] so that any path of length $D$ will be shortcut to expected length $\frac{D}{10}$ in one run of $\ParallelSC(G, k, 0, 0)$.

\paragraph{Setup for proof of \cref{par:paradiambound}.} An execution of $\ParallelSC(G', k, r, r^\fringe)$ will generate many recursive subproblems, each of the form $\ParallelSC(G'[V^\ring_v], k, r, r^\fringe+1)$ or \\ $\ParallelSC(G'[V_i], k, r+1, 0).$ These recursive subproblems (other than the fringe vertices) all involve disjoint sets of vertices, hence the path $P$ that we are trying to analyze gets ``split" into subpaths when we look at the recursive subproblems. Our next claim gives structure on how we can split up the path $P$ into subpaths, which are contained in our various subproblems. Here, we define $s(P, G')$ to be the number of partially path-related vertices to $P$ in $G'$.

\begin{lemma}[Splitting up a path]
\label{par:parapathsplit}
Consider the setup described in the above paragraph and \cref{par:paradiambound}. Let $\kappa$ be such that $\kappa_{2r+1} \le \kappa \le \kappa_{2r}$ and $S \subseteq R_{G'}(P, \kd)$ be the set of shortcutters. There exists paths $P^{S, \kappa}_i$ (possibly empty) for $1 \le i \le |S|+1$ and $P^{S, \kappa, \fringe}_i$ (possibly empty) for $1 \le i \le |S|$ such that
\begin{enumerate}
\item If a vertex $u \in S$ is a $\kd$-bridge, then all $P^{S, \kappa}_i$ and $P^{S, \kappa, \fringe}_i$ are empty.
\item If no vertex $u \in S$ is a $\kd$-bridge, then the vertex disjoint union of all $P^{S, \kappa}_i$ for $1 \le i \le |S|+1$ and $P^{S, \kappa, \fringe}_i$ for $1 \le i \le |S|$ is exactly $P$.
\item Each $P^{S, \kappa}_i$ is inside $V_{a_i}$ for some index $a_i$.
\item Each $P^{S, \kappa, \fringe}_i$ is a subset of $V^\ring_{u_i}$ for some vertex $u_i \in S$ (line \ref{line:fringerecurse}).
\end{enumerate}
Additionally, we have that \begin{equation}\label{eq:normalbound} \E_{\kappa, |S| = t} \left[\sum_{i=1}^{t+1} s(P^{S, \kappa}_i, G'[V_{a_i}])\right] \le \frac{2}{t+1} s(P, G') \end{equation} and  \begin{equation}\label{eq:fringebound} \E_{\kappa,S} \left[\sum_{i=1}^t s(P^{S, \kappa, \fringe}_i, G'[V^\ring_{u_i}]) \right] \le \frac{1}{1000 k \log^3 n} s(P, G'), \end{equation} where $\E_{\kappa, |S| = t}$ refers first selecting $\kappa \in [\kappa_{2r+1}, \kappa_{2r}]$ uniformly at random and then picking the shortcutter set $S$ as a random subset of $R_{G'}(P, \kd)$ of size exactly $t$ (as opposed to the way $S$ is selected in \cref{algo:parallelsc}). $\E_{\kappa, S}$ refers to first selecting $\kappa \in [\kappa_{2r+1}, \kappa_{2r}]$ uniformly at random and then selecting $S$ as in \cref{algo:parallelsc}, where every vertex $u \in V(G')$ is in $S$ with probability $p_r.$
\end{lemma}
Here, we assert that $S \subseteq R_{G'}(P, \kd)$ as searching to distance $\kd$ from vertices $u \not\in R_{G'}(P, \kd)$ does not affect the path $P$.
\begin{proof}
We describe a process for constructing the $P^{S, \kappa}_i$ and $P^{S, \kappa, \fringe}_i$, and then verify that this construction satisfies all the necessary constraints.

If a vertex $u \in S$ is a $\kd$-bridge, then we set all the $P^{S, \kappa}_i$ and $P^{S, \kappa, \fringe}_i$ to be empty. Otherwise, first consider the case where $u \in S$ is a $\kd$-descendant. Let $v_0 \pe v_1 \pe \cdots \pe v_\ell$ be the vertices on the path $P$. Define \[ \bot(u) = \min \{i : 0 \le i \le \ell, v_i \in \RA_{G'}(u, \kd) \} \text{ and }\] \[ \top(u) = \min\left(\bot(u)+1, \max \{i : 0 \le i \le \ell, v_i \in \RA_{G'}(u, (\kappa-1)D) \} \right) \] and $\top(u) = \bot(u)+1$ in the case that the set $\{i : 0 \le i \le \ell, v_i \in \RA_{G'}(u, (\kappa-1)D) \}$ is empty. Define \[ W^\top(u) = \{ v_i : 0 \le i < \top(u) \}, W^\bot(u) = \{ v_i : \bot(u) < i \le \ell \}, W^\fringe(u) = \{ \top(u) \le i \le \bot(u) \}. \] In particular, it is easy to verify that $W^\fringe(u) \subseteq V^\fringe_u.$ For the case where $u \in S$ is a $\kd$-ancestor, define \[ \top(u) = \max \{i : 0 \le i \le \ell, v_i \in \RD_{G'}(u, \kd) \} \text{ and } \] \[ \bot(u) = \max\left(\top(u)-1, \min \{i : 0 \le i \le \ell, v_i \in \RD_{G'}(u, (\kappa-1)D) \}\right), \] where $\bot(u) = \top(u)-1$ in the case that the set $\{i : 0 \le i \le \ell, v_i \in \RD_{G'}(u, (\kappa-1)D) \}$ is empty. We can now define $W^\top(u), W^\bot(u)$, and $W^\fringe(u)$ analogously.

Now, our goal is to define a set of natural subpaths of $P$ which are inside various $V_i$ (line \ref{line:vi}). Specifically, for a subset $T \subseteq S$, we define \[ W_T = \left(\bigcap_{u \in T} W^\top(u) \right) \cap \left(\bigcap_{u \in S\bs T} W^\bot(u) \right). \] $W_T$ was defined so that it is inside some subset $V_i$ (line \ref{line:vi}) and is a subpath of $P$. To show that $W_T$ is contained inside some $V_i$ it suffices to show that all vertices in $W^\top(u)$ receive the same label from $u$ (and by symmetry the same holds for $W^\bot(u)$). In the case where $u$ is a $\kd$-descendant of $P$, it is easy to verify that $W^\top(u) \subseteq \RA_{G'}(u, \kd)$, and that all vertices in $W^\bot(u)$ are unrelated to $u$. In the case where $u$ is a $\kd$-ancestor, it is easy to verify that $W^\bot(u) \subseteq \RD_{G'}(u, \kd)$, and that all vertices in $W^\top(u)$ are unrelated to $u$. This shows that $W_T$ always lies inside $V_i$ for some $i$.

It is easy to check that the number of nonempty $W_T$ over subsets $T \subseteq S$ is at most $|S|+1$. We denote these sets  $W_1, W_2, \dots, W_{|S|+1}.$ For simplicity, label the paths $W^\fringe(u)$ for $u \in S$ as $W_{|S|+2}, \dots, W_{2|S|+1}.$ By definition, we see that $P$ itself is the (not necessarily disjoint) union of $W_1, W_2, \dots, W_{2|S|+1}$ . Now, one can easily check that for any paths $W_1, W_2, \dots, W_{2|S|+1}$ whose union is $P$, then there are subpaths $Z_1, Z_2, \dots, Z_{2|S|+1}$ (possibly empty) such that $Z_i \subseteq W_i$ and that the vertex disjoint union of the $Z_i$ is $P$.

Now, we define the $P^{S, \kappa}_i$ as the subpaths $Z_j$ for $1 \le j \le |S|+1$, and define the $P^{S, \kappa, \fringe}_i$ as the subpaths $Z_j$ for $|S|+2 \le j \le 2|S|+1$, which then by definition were subpaths of $W^\fringe(u)$ for some $u$. Now we verify that these choices satisfy the conditions of \cref{par:parapathsplit}.

\paragraph{Construction satisfies item 1.} This is by definition (see the first sentence of the proof).

\paragraph{Construction satisfies item 2.} We have defined the $Z_j$ so that their vertex disjoint union is $P$, and each of the $P^{S, \kappa}_i$ and $P^{S, \kappa, \fringe}_i$ is one of the $Z_j$ or empty.

\paragraph{Construction satisfies item 3.} The definition \[ W_T = \left(\bigcap_{u \in T} W^\top(u) \right) \cap \left(\bigcap_{u \in S\bs T} W^\bot(u) \right) \] corresponds to assigning the labels to vertices $V(G')$ and partitioning into sets $V_i.$ As each $P^{S, \kappa}_i$ was a subpath of some $W_T$, the claim follows.

\paragraph{Construction satisfies item 4.} Each $P^{S, \kappa, \fringe}_i$ is a subpath of one of $W_{|S|+2}, \dots, W_{2|S|+1}$, which were each $W^\fringe(u)$ for some $u$. We have noted that $W^\fringe(u) \subseteq V^\ring_u.$

\paragraph{Construction satisfies \cref{eq:fringebound}.} Note that \begin{equation}\label{eq:kanyewest} \E_{\kappa,S} \left[\sum_{i=1}^t s(P^{S, \kappa, \fringe}_i, G'[V^\ring_{u_i}]) \right] \le \E_{\kappa, S} \left[\sum_{v \in S} \left|V_v^\ring \cap R_{G'}(P, \kappa_{2r}D)\right| \right] \end{equation} as no vertices outside of $R_{G'}(P, \kappa_{2r} D)$ can become path related. The right hand side of \cref{eq:kanyewest} is counting for every vertex $u$ in $R_{G'}(P, \kappa_{2r}D)$ the expectation over the choices of $\kappa, S$ of the number of times that $u$ is in a fringe subproblem, i.e. $u \in V^\ring_v$ for some $v \in S.$ By applying \cref{par:singlefringebound}, we can see that this expression in \cref{eq:kanyewest} is with high probability at most \[ \E_{\kappa, S} \left[\sum_{v \in S} \left|V_v^\ring \cap R_{G'}(P, \kappa_{2r}D)\right| \right] \le \frac{1}{1000 k \log^3 n} |R_{G'}(P, \kappa_{2r}D)| = \frac{1}{1000 k \log^3 n} s(P, G'). \]

\paragraph{Construction satisfies \cref{eq:normalbound}.} The proof essentially follows the same shape as \cref{seq:2k+1lemma} and \cref{seq:1llemma}. Our main claim is the following.
\begin{claim}
\label{par:claim1}
\label{par:claim1}
Consider an arbitrary subset of vertices $\{u_1, u_2, \cdots, u_{t+1} \} \subseteq R_{G'}(P, \kd)$. Then there are at most $2$ indices $1 \le h \le t+1$ such that setting $S = \{u_j : j \neq h \}$ and running the procedure described so far in the proof of \cref{par:parapathsplit} to produce the sets $P^{S, \kappa}_i$ and $P^{S, \kappa, \fringe}_i$ results in $u_h$ still being path related in a non-fringe subproblem, i.e. $u_h \in R_{G'[V_{a_i}]}(P^{S, \kappa}_i, \kd)$ for some $1 \le i \le t+1.$
\end{claim}
Now we explain why \cref{par:claim1} implies \cref{eq:normalbound}. Intuitively, in \cref{eq:normalbound}, we can think of choosing the set $S$ with $|S| = t$ instead as first uniformly randomly choosing a subset $S' \subseteq R_{G'}(P, \kd)$ with $|S'| = t+1$ and then choosing uniformly randomly choosing $S \subseteq S'$ with $|S| = t.$ We then apply \cref{par:claim1}. Formally, we have the following, where $\one(S, u)$ is the function evaluating to $1$ when $u \in R_{G'[V_{a_i}]}(P^{S, \kappa}_i, \kd)$ for some $1 \le i \le t+1$ and $0$ otherwise.
\begin{align*}
&\E_{\kappa, |S| = t} \left[\sum_{i=1}^{t+1} s(P^{S, \kappa}_i, G'[V_{a_i}])\right] \\
&= \E_{\kappa, |S| = t} \left[\sum_{u \in R_{G'}(P, \kd)} \one(S, u)\right] \\
&\le s(G', P) \cdot \E_{\kappa, |S'| = t+1} \E_{\substack{S \subset S' \\ |S| = t}} \left[\one(S'\bs S, S) \right] \\
&\le \frac{2}{t+1} s(G', P)
\end{align*}
by \cref{par:claim1} as desired. Between the second and third lines we have used linearity of expectation on the inner sum and symmetry. For the remainder of this proof, we focus on proving \cref{par:claim1}. We instead prove a slightly stronger and more restrictive version of \cref{par:claim1} in the case that all the $u_i$ are path ancestors (or all path descendants).
\begin{claim}
\label{par:claim2}
Consider an arbitrary subset of vertices $\{u_1, u_2, \cdots, u_{t+1} \} \subseteq R_{G'}(P, \kd)$ such that the $u_i$ are all path ancestors or all path descendants. Then there is at most one index $1 \le h \le t+1$ such that setting $S = \{u_j : j \neq h \}$ and running the procedure described so far in the proof of \cref{par:parapathsplit} to produce the sets $P^{S, \kappa}_i$ and $P^{S, \kappa, \fringe}_i$ results in $u_h$ still being path related in a non-fringe subproblem, i.e. $u_h \in R_{G'[V_{a_i}]}(P^{S, \kappa}_i, \kd)$ for some $1 \le i \le t+1.$
\end{claim}
We now explain why \cref{par:claim1} follows from \cref{par:claim2} and then show \cref{par:claim2}. To show \cref{par:claim1}, note that the claim is trivial if the set $S' = \{u_1, u_2, \cdots, u_{t+1}\}$ contains a bridge. Indeed, as picking the set $S$ to contain a bridge leaves no path related subproblems, the only index $j$ satisfying the constraints of \cref{par:claim1} must be the $j$ where $u_j$ is the bridge. Now, assume that all of the vertices in $S'$ are either path ancestors or descendants (not bridges). Applying \cref{par:claim2} on the subset of path ancestors among $S'$ and the subset of path descendants among $S'$ immediately implies \cref{par:claim1}: we get one valid index from the path ancestors and one from the descendants, for a total of two.

We now show \cref{par:claim2}. Consider the set $S' = \{u_1, u_2, \cdots, u_{t+1}\}$, where all the $u_i$ are path ancestors. Therefore, for any subset $S \subseteq S'$, if we shortcut using vertices in $S$, then all vertices in $P$ will only get labels of the form $u_i^\Des$, and no $u_i^\Anc$ labels. For every vertex $u \in S'$, define \[ \alpha(u) = \min \{i : 0 \le i \le \ell \text{ and } v_j \in R_{G'}^\Des(u, \kd) \forall j \ge i \}, \] where we recall that the path $P$ consists of vertices $v_0 \pe v_1 \pe \cdots \pe v_\ell.$ We claim that the only index $j$ that could satisfy the condition in \cref{par:claim2} is such that
\begin{itemize}
\item We have that $u_j \not\pe^{\kd} u_{j'}$ for all $j' \ne j$.
\item $\alpha(u_j)$ is unique and minimal out of all the values of $\alpha(u_{j'})$ for the vertices $u_{j'}$ satisfying the first point.
\end{itemize}
Indeed, for an index $j$, if there is an index $j'$ such that $u_j \pe^{\kd} u_{j'}$ then vertex $u_j$ receives the label $u_{j'}^\Anc$ when we shortcut from $u_{j'}$. On the other hand, no vertices on the path $P$ receive the label $u_{j'}^\Anc$ as $u_{j'}$ is a path ancestor. Now, let \[ I = \{j : u_j \not\pe^{\kd} u_{j'} \forall j' \neq j \}, \] the set of indices which satisfied the first condition. Note that by definition, for any distinct indices $j, j' \in I$, the vertices $u_j$ and $u_{j'}$ are unrelated within distance $\kd$. Now, assume that $\alpha(u_j) \ge \alpha(u_{j'})$. We aim to show that in this case that $u_j \not\in R_{G'[V_{a_i}]}(P^{S, \kappa}_i, \kd)$ for any $1 \le i \le t+1$, which would complete the proof. Indeed, if $u_j \in R_{G'[V_{a_i}]}(P^{S, \kappa}_i, \kd)$, then by the definition of $\alpha(u_j)$ we would need for $P^{S, \kappa}_i$ to be a subset of the path $P$ from $v_{\alpha(u_j)}$ to $v_\ell.$ Otherwise, the path $P^{S, \kappa}_i$ would be forced to be part of a fringe subproblem induced by $u_j$. Now, again by the definition of $\alpha(u_{j'})$, we know that all vertices in the path $P$ from $v_{\alpha(u_j)}$ to $v_\ell$ will receive a label of $u_{j'}^\Des$ because $\alpha(u_j) \ge \alpha(u_{j'})$, while vertex $u_j$ does not receive a label from $u_{j'}$, as $u_{j'} \not\pe^{\kd} u_j.$ Therefore, $u_j \not\in R_{G'[V_{a_i}]}(P^{S, \kappa}_i, \kd)$ for any $1 \le i \le t+1$, which completes the proof.

\end{proof}
Now, we prove \cref{par:paradiambound}.
\begin{proof}[Proof of \cref{par:paradiambound}]
The proof is via induction and \cref{par:parapathsplit}. Indeed, the claim is trivial for $\br = 0$, as \cref{algo:parallelsc} chooses all vertices as shortcutters as $p_r = 1.$

We first claim that (using the same notation as \cref{par:parapathsplit}) that with probability $1-n^{-10}$ we have that $|R_{G'}(P, \kd)| \le 100k\log n.$ Indeed, note that \[ |R_{G'}(P, \kd)| \le |R_{G'}(\head(P), (\kappa+1)D)| + |R_{G'}(\tail(P), (\kappa+1)D)| \le 100k\log n\] with probability at least $1 - n^{-10}$ by \cref{par:paralabelbound}.

For a path $P$, parameter $\kappa$, and set $S$ we have by \cref{par:parapathsplit} that we can partition $P$ into disjoint paths $P^{S, \kappa}_i$ and $P^{S, \kappa, \fringe}_i$. Then the final length of $P$ after shortcutting through $S$ and recursive shortcutting is clearly at most the total length of all the $P^{S, \kappa}_i$ and $P^{S, \kappa, \fringe}_i$ after shortcutting, plus at most $2|S|+2$ (for edges between adjacent paths). Note that this even holds in the case where $S$ contains a bridge: in this case all the $P^{S, \kappa}_i$ and $P^{S, \kappa, \fringe}_i$ are empty by definition, and the shortcutted length of $P$ is $2$, which is at most $2|S|+2.$

By induction, we have that the expected length of $P$ after shortcutting through $S$ and recursive shortcutting is at most
\begin{align} \E_{\kappa, S}\Big[2|S| + 2 &+ \sum_{i=1}^{|S|+1} 5\br \left( \sqrt{2} + \frac{1}{2\log n}\right)^{\br-1} s(P^{S, \kappa}_i, G'[V_{a_i}])^\frac12  \\ &+ \sum_{i=1}^{|S|} 5(\br+1) \left( \sqrt{2} + \frac{1}{2\log n}\right)^{\br} s(P^{S, \kappa, \fringe}_i, G'[V_{u_i}^\ring])^\frac12 \Big]
\\ \label{eq:finalish1} &= 2 + 2\E_{\kappa, S}[|S|] + 5\br \left( \sqrt{2} + \frac{1}{2\log n}\right)^{\br-1} \E_{\kappa, S}\Big[ \sum_{i=1}^{|S|+1} s(P^{S, \kappa}_i, G'[V_{a_i}])^\frac12\Big]
\\ \label{eq:finalish2} &+ 5(\br+1) \left( \sqrt{2} + \frac{1}{2\log n}\right)^{\br} \E_{\kappa, S} \Big[ \sum_{i=1}^{|S|} s(P^{S, \kappa, \fringe}_i, G'[V_{u_i}^\ring])^\frac12 \Big]
\end{align} by induction, where $\E_{\kappa, S}$ refers to first selecting $\kappa \in [\kappa_{2r+1}, \kappa_{2r}]$ uniformly at random and then selecting $S$ analogous to \cref{algo:parallelsc}, where every vertex $u \in R_{G'}(P, \kd)$ is in $S$ with probability $p_r.$ We only care about those $u \in R_{G'}(P, \kd)$ as those are the only path related vertices, and picking other vertices can only make less vertices path related in the future.

We bound the pieces of \cref{eq:finalish1} and \cref{eq:finalish2} now. We split the analysis into cases depending on the size of $\E_{\kappa, S}[|S|].$

\paragraph{Case 1: $p_r \ge s(P, G')^{-\frac12}$.} We assume that the length $\ell$ of the path $P$ must satisfy \[ \ell \ge 5(\br+1) \left( \sqrt{2} + \frac{1}{2\log n}\right)^{\br} s(P, G')^\frac12, \] or else the claim is trivially true. Now, we have that the probability that some vertex of $P$ is in $S$ is at least \[ 1 - (1-p_r)^\ell \ge 1 - \exp(-\ell p_r) \ge 1-n^{-5} \] by our assumed bounds on $\ell$ and $p_r.$ Therefore, our path $P$ is shortcut to length $2$ with high probability. This completes the analysis for this case.

\paragraph{Case 2: $p_r \le s(P, G')^{-\frac12}$.} 

We now bound each of the pieces of \cref{eq:finalish1} and \cref{eq:finalish2}.
\paragraph{Bound on $\E_{\kappa, S}[|S|]$.} Note that \[ \E_{\kappa, S}[|S|] \le p_r s(P, G') \le s(P, G')^\frac12 \] by our assumed bound on $p_r.$
\paragraph{Bound on rightmost term in \cref{eq:finalish1}.} Using the Cauchy-Schwarz inequality and \cref{par:parapathsplit} \cref{eq:normalbound} we can see that
\begin{align*}&5\br \left( \sqrt{2} + \frac{1}{2\log n}\right)^{\br-1} \E_{\kappa, S}\Big[ \sum_{i=1}^{|S|+1} s(P^{S, \kappa}_i, G'[V_{a_i}])^\frac12\Big] \\
&\le 5\br \left( \sqrt{2} + \frac{1}{2\log n}\right)^{\br-1} \sum_t \Pr[|S|=t] \cdot \E_{\kappa, |S|=t} \Big[ \sum_{i=1}^{t+1} s(P^{S, \kappa}_i, G'[V_{a_i}])^\frac12\Big] \\
&\le 5\br \left( \sqrt{2} + \frac{1}{2\log n}\right)^{\br-1} \sum_t \Pr[|S|=t] \cdot \E_{\kappa, |S|=t} \left[ \left((t+1)\sum_{i=1}^{t+1} s(P^{S, \kappa}_i, G'[V_{a_i}])\right)^\frac12\right] \\
&\le 5\br \left( \sqrt{2} + \frac{1}{2\log n}\right)^{\br-1} \sum_t \Pr[|S|=t] \cdot \E_{\kappa, |S|=t} \left[ \left((t+1)\sum_{i=1}^{t+1} s(P^{S, \kappa}_i, G'[V_{a_i}])\right)\right]^\frac12 \\
&\le 5\br \left( \sqrt{2} + \frac{1}{2\log n}\right)^{\br-1} \sum_t \Pr[|S|=t] \cdot \sqrt{2} \cdot s(P, G')^\frac12 \\
&= 5\br \left( \sqrt{2} + \frac{1}{2\log n}\right)^{\br-1} \cdot \sqrt{2} \cdot s(P, G')^\frac12.
\end{align*}
\paragraph{Bound on \cref{eq:finalish2}.} We follow the same plan as in the above paragraph. Using the Cauchy-Schwarz inequality, \cref{par:parapathsplit} \cref{eq:fringebound}, and our observation above that $|R_{G'}(P, \kappa_{2r}D)| \le 100k\log n$ with high probability we can see that
\begin{align*}&5(\br+1) \left( \sqrt{2} + \frac{1}{2\log n}\right)^{\br-1} \E_{\kappa, S} \Big[ \sum_{i=1}^{|S|} s(P^{S, \kappa, \fringe}_i, G'[V_{u_i}^\ring])^\frac12 \Big] \\
&\le 5(\br+1) \left( \sqrt{2} + \frac{1}{2\log n}\right)^{\br-1} \E_{\kappa, S} \left[ \left(|S| \sum_{i=1}^{|S|} s(P^{S, \kappa, \fringe}_i, G'[V_{u_i}^\ring])\right)^\frac12 \right] \\
&\le 5(\br+1) \left( \sqrt{2} + \frac{1}{2\log n}\right)^{\br-1} \E_{\kappa, S} \left[ \left(|S| \sum_{i=1}^{|S|} s(P^{S, \kappa, \fringe}_i, G'[V_{u_i}^\ring])\right)\right]^\frac12 \\
&\le 5(\br+1) \left( \sqrt{2} + \frac{1}{2\log n}\right)^{\br} \sum_t \Pr[|S|=t] \cdot \sqrt{\frac{|R_{G'}(P, \kappa_{2r}D)|}{1000 k\log^3 n}} \cdot s(P, G')^\frac12\\
&\le 5(\br+1) \left( \sqrt{2} + \frac{1}{2\log n}\right)^{\br} \sum_t \Pr[|S|=t] \cdot \sqrt{\frac{100k\log n}{1000 k\log^3 n}} \cdot s(P, G')^\frac12 \\
&\le 5\br \left( \sqrt{2} + \frac{1}{2\log n}\right)^{\br-1} \cdot \frac{1}{2\log n} \cdot s(P, G')^\frac12.
\end{align*}
Summing all our contributions gives that the expression in \cref{eq:finalish1} and \cref{eq:finalish2} is at most
\begin{align*}
&2 + 2\E_{\kappa, S}[|S|] + 5\br \left( \sqrt{2} + \frac{1}{2\log n}\right)^{\br-1} \E_{\kappa, S}\Big[ \sum_{i=1}^{|S|+1} s(P^{S, \kappa}_i, G'[V_{a_i}])^\frac12\Big]
\\ &+ 5(\br+1) \left( \sqrt{2} + \frac{1}{2\log n}\right)^{\br} \E_{\kappa, S} \Big[ \sum_{i=1}^{|S|} s(P^{S, \kappa, \fringe}_i, G'[V_{u_i}^\ring])^\frac12 \Big]
\\ &\le 2 + 2 s(P, G')^\frac12 + 5\br \left( \sqrt{2} + \frac{1}{2\log n}\right)^{\br-1} \cdot \sqrt{2} \cdot s(P, G')^\frac12 + 5\br \left( \sqrt{2} + \frac{1}{2\log n}\right)^{\br-1} \cdot \frac{1}{2\log n} \cdot s(P, G')^\frac12
\\ &\le 5(\br+1) \left( \sqrt{2} + \frac{1}{2\log n}\right)^{\br} s(P, G')^\frac12
\end{align*} as desired.
\end{proof}
To conclude the proof of the correctness of \cref{algo:paralleldiam}, it suffices to argue after a run \cref{algo:paralleldiam}, for any vertices $s, t$ such that $s \pe t$, after shortcutting there is now a path of length at most $D$ from $s$ to $t$.
\begin{lemma}
\label{par:masterlemma}
Perform an execution of $\ParallelDiam(G, k)$ on an $n$-node, $m$-edge digraph $G$. With high probability for any vertices $s, t$ such that $s \pe t$ there is a path of length at most $D$ from $s$ to $t$ in $G$, where $D$ is as in \cref{algo:parallelsc}.
\end{lemma}
\begin{proof}
By \cref{par:paradiambound} and the observation that \[ 5(\log_k n+1) \left( \sqrt{2} + \frac{1}{2\log n}\right)^{\log_k n} n^\frac12 \le 10\log n \sqrt{2}^{\log_k n} n^\frac12 \le \frac{D}{10}, \] we can see that for a path $P$ of length at most $P$ the probability that it is shortcut to length at most $\frac{D}{2}$ in a run of $\ParallelSC(G, k, 0, 0)$ which does not exceed the work or shortcut edge bound by ten times is at least $1 - \frac{1}{5} - \frac{2}{10} > \frac12$ by Markov's inequality. Therefore, running $\ParallelSC(G, k, 0, 0)$ $10\log n$ times as in \cref{algo:paralleldiam} shortcuts a path of length $D$ to length at most $\frac{D}{2}$ with high probability. Doing this $10\log n$ times then guarantees that an arbitrary path is shortcut to length at most $D$ with high probability. As we only need to shortcut $O(n^2)$ paths (a single path for every pair of vertices $s, t$) this is sufficient.
\end{proof}
Combining \cref{par:finalparabound} and \cref{par:masterlemma} readily gives a proof of \cref{thm:parallel}. Combining these along with a breadth first search to depth $D$ now gives a proof of \cref{thm:main_parallel}. 
\begin{proof}[Proof of \cref{thm:main_parallel}]
Take $k = O(\log n)$.
For input graph $G$, consider first performing an execution of $\ParallelDiam(G, k)$ and then running a breadth first search from vertex $s$. Clearly this algorithm solves the single source reachability problem from $s$. \cref{par:finalparabound} gives with high probability the desired $\O(mk + nk^2)$ total work bound and shows that the algorithm adds at most $\O(nk)$ shortcuts with high probability. \cref{par:masterlemma} shows that the breadth first search runs in depth $D \le \O\left(n^{\frac12 + O\left(\frac{1}{\log k}\right)}\right)$, where $D$ is as in \cref{algo:parallelsc}. By \cref{par:finalparabound} we have that an execution of $\ParallelDiam(G, k)$ can be implemented in parallel depth $\O\left(\poly(k) \cdot n^{\frac12 + O\left(\frac{1}{\log k}\right)}\right).$ This finishes the proof.
\end{proof}

\section{Distributed Single Source Reachability in the \congest~Model}
\label{sec:distributed}

In this section we explain how to combine our parallel algorithm in \cref{sec:parallel} with the techniques of \cite{Nan14} and \cite{GU15} to obtain new algorithms for the single source reachability problem in the \congest~model. In \cref{sec:congestprelim} we formally define the \congest~model and the single source reachability problem, as well as covering some standard results in the \congest~model. In \cref{sec:congestalgo} we state our main algorithm \DistrReach~for solving the single source reachability problem. Finally, in \cref{sec:congestanalysis} we sketch a proof of how to apply the results of \cref{sec:parallel} to bound the number of rounds our algorithm takes.

\subsection{Preliminaries}
\label{sec:congestprelim}
\paragraph{The \congest~model.} In the \congest~model, we model a communication network as an undirected unweighted graph $G$ with $n$ vertices and $m$ edges, where the processors are modeled as vertices and edges as bounded bandwidth links between processors. We let $V(G)$ and $E(G)$ denote the vertex set of $G$ and the (directed) edge set of $G$ respectively. The processors/vertices are assumed to have distinct polynomially bounded IDs. Each vertex has infinite computational power, but only knows little about the graph $G$; it only knows the IDs of its neighbors and no other information about the graph $G$.

Now we define the \emph{single source reachability problem} that we consider. Here, each vertex also knows which of its neighbors are in-neighbors and which are out-neighbors in the directed graph $G$. Here, communication on edges are bidirectional even though the edges themselves are directed edges. Additionally, there is a single source vertex $s$. As vertex $s$ can communicate that it is the source to all other vertices in $O(D)$ rounds (\cref{lemma:broadcast}), we assume that all vertices know $s$ at the start. The goal of the single source reachability problem is for every vertex $t \in V(G)$ to know whether is an $s\to t$ path in $G$ at the end of the algorithm. In particular, $s$ doesn't have to know whether there is an $s\to t$ path in $G$, only $t$ must know.

The performance of an algorithm or protocol is judged by the number of \emph{rounds} of distributed communication in the worst case. At the start of each round, every vertex can send a message of length $O(\log n)$ to each of its neighbors, where the message need not be the same across different neighbors. The messages arrive at the end of the round. Running times are analyzed in terms of the number of vertices $n$ and the undirected/hop diameter of $G$, which we denote as $D$. As every vertex can learn $n$ within $O(D)$ rounds, we assume that all vertices know $n$.

We now note a few standard results about distributed computation in the \congest~model.
\begin{lemma}[Global message broadcast protocol, \cite{Pel00}]
\label{lemma:broadcast}
Suppose each $v \in V$ holds $k_v \ge 0$ messages of $O(\log n)$ bits each, for a total of $K = \sum_{v \in V} k_v$.
Then all nodes in the network can receive these $K$ messages within $O(K + D)$ rounds.
\end{lemma}
\begin{theorem}[\!\!\cite{Gha15}]
Consider $k$ distributed algorithms $A_1, \dots, A_k.$ Let \dilation~be such that each algorithm $A_i$ finishes in \dilation~rounds if it runs individually. Let \congestion~be such that there are at most \congestion~messages, each of size $O(\log n)$, sent
through each edge (counted over all rounds), when we run all algorithms together.
There is a distributed algorithm that can execute $A_1, \dots, A_k$ in $O(\dilation + \congestion \cdot \log n)$ rounds in the \congest~model.
\end{theorem}
\begin{corollary}[Low distance breadth first search]
\label{lemma:bfs}
Consider a set $T \subseteq V(G)$ and a distance $h$. There is a protocol such that for every vertex $v$, it learns the subsets of $T$ that are distance $h$ to and from itself in $\O(|T| + h)$ rounds.
\end{corollary}

\subsection{Algorithm description}
\label{sec:congestalgo}
\paragraph{Algorithm overview.} Our algorithm consists of multiple steps. The first is a standard procedure reminiscent of the classic $O(m\sqrt{n})$ time shortcutting algorithm of Ullman and Yannakakis \cite{UY91}. It has been used in the earlier works of Ghaffari and Udwani \cite{GU15} and Nanongkai \cite{Nan14}. We randomly sample each vertex in $V(G)$ with probability $\frac{10\alpha \log n}{n}$ to be in a subset $T \subseteq V(G).$ We also assert that our source $s \in T$. With high probability $|T| = \O(\alpha).$ Using \cref{lemma:bfs} with parameter $h = \O(n/\alpha)$, for each vertex $v$, it learns exactly the subset of $T$ that can reach itself within distance $h$, and the subset of $T$ that it can reach within distance $h$. This takes $\O(\alpha + h) = \O(\alpha + \frac{n}{\alpha})$ rounds. Now, we build the \emph{skeleton graph} $G_T$ on $T$ using the reachability relations learned from the breadth first searches. Concretely, for vertices $u, v \in T$, there is an edge $u \to v$ in graph $G_T$ if $d_G(u, v) \le h.$ It is direct to see that with high probability that the reachability from $s$ to $T$ on $G_T$ is the same as that on $G$.

Now, every vertex in $T$ tries to learn whether it is reachable from $s$ in graph $G_T$. After this, the vertices in $T$ can global broadcast whether they are reachable from $s$ using \cref{lemma:broadcast} in $O(D + |T|)$ rounds. Then every vertex in $G$ can locally determine whether $s$ can reach it, as desired.

In order to let every vertex in $T$ learn whether it is reachable from $s$, we run a small variation on \cref{algo:parallelsc} and \cref{algo:paralleldiam} on $G_T$. Intuitively, during every unit of time in the parallel computation every vertex globally broadcasts much of its computational transcript, such as its current labels, which new shortcut edges are being added, whether it has just been visited in a breadth first search, etc. Through the guarantees of \cref{thm:parallel} with $k = \log |T|$ we can see that this process essentially takes $\O(D|T|^{1/2 + o(1)} + |T|) = \O(D\alpha^{1/2 + o(1)} + \alpha)$ rounds. The first term comes from a fixed cost of $O(D)$ to globally broadcast during each stage of a $|T|^{1/2 + o(1)}$ depth algorithm. The second term is the total number of messages sent by all vertices during the process, which corresponds to the nearly linear runtime of the algorithm. Combining this with the first paragraph gives an algorithm taking \[ \O\left(\frac{n}{\alpha} + \alpha + D\alpha^{1/2 + o(1)}\right) \] rounds. Properly trading off parameters then gives \cref{thm:congest}.

We remark that using the reachability algorithm of \cite{Fine18} (instead of our improved algorithm) would also imply a new result; precisely the number of rounds would be $\O(\sqrt{n} + n^{2/5} D^{3/5})$, which is nearly optimal for $D \le n^{1/6}$.

We now give algorithm \DistrReach~(\cref{algo:congest}). It is fairly clear that in this way, \cref{algo:congest} can simulate \ParallelDiam~in the \congest~model.

\begin{algorithm}[h!]
\caption{\DistrReach$(G, s)$, distributed algorithm for solving the single source reachability problem from $s$ in $G$ in the \congest~model.}
\begin{enumerate}
\item Choose $T \subseteq V(G)$ as follows: each vertex $v \in V(G)$ has a $\frac{10\alpha \log n}{n}$ probability of being in $T$. Set $h = \frac{10n \log n}{\alpha}.$ Apply \cref{lemma:bfs} so that all vertices $v \in V(G)$ learn what vertices in $T$ are within distance $h$ to / from itself. Build the skeleton graph $G_T$ with vertex set $T$ as follows: for $u, v \in T$ there is an edge $u \to v$ if and only if $d_G(u, v) \le h$. \label{line:distline1}
\item Run a variant of \ParallelDiam$(G, k)$~on $G_T$. After every unit of parallel depth, every vertex globally broadcasts all new computation of the following forms, as well as an ID for the current recursive subproblem it is broadcasting for. Note that a vertex can broadcast / do computation for multiple subproblems simultaneously. \label{line:distline2}
\begin{itemize}
\item Vertex $v$ is in the shortcutter set $S$.
\item Vertex $v$ is visited by a breadth first search from a vertex $u \in S$.
\item Vertex $v$ gets a new label / becomes a fringe vertex.
\item A new shortcut edge is added to the graph.
\end{itemize}
In addition, each subproblem must choose a ``leader" to pick the $\kappa$ parameter for that subproblem.

The skeleton graph $G_T$ is updated after each run of $\ParallelSC(G, k, 0, 0)$ to include all the new shortcut edges.
\item Simulate a BFS on $G_T$ (with all the new shortcut edges) of depth $100\sqrt{2}^{\log_k n} n^\frac12 \log^2 n$ so that all $v \in T$ learn whether they are reachable from $s$. All vertices $v \in T$ globally broadcast whether they are reachable from $s.$ Each vertex $v \in V(G)$ now locally computes whether it is reachable from $s$.
\end{enumerate}
\label{algo:congest}
\end{algorithm}

\subsection{Analysis}
\label{sec:congestanalysis}
The bulk of the analysis follows in the same way as the analysis in \cref{sec:parallel}.
\begin{theorem}
\label{thm:congestmain}
Algorithm \DistrReach~when executed on an $n$-vertex $m$-edge graph $G$ with high probability solves the single source reachability problem in $\O\left(\alpha + (n/\alpha) + D\alpha^{1/2 + o(1)} \right).$
\end{theorem}
\begin{proof}
By \cref{lemma:bfs}, we have that Line~\ref{line:distline1}  in Algorithm \DistrReach~can be implemented with $O(|T| + h) = \O\left(\alpha + n/\alpha\right)$ rounds.

Our procedures simulates running $\ParallelDiam(G_T,k)$ over the graph. After every unit of parallel depth in our algorithm, we globally broadcast the complete reachability information of all performed searches over the graph. This consists of $\otilde(\alpha)$ messages in total. Further, we observe that during each graph traversal we perform during $\ParallelDiam$~nodes will locally check if one of their neighbors have been marked by a search and broadcast themselves appropriately if they have not yet been marked by the same search. Thus every node communicates one message for every search it is visited in: this is just the total number of shortcuts, and  the total amount of communication needed to perform all the searches of $\ParallelDiam$~is $\O(\alpha)$. Now $\ParallelDiam$~itself takes $\O(\alpha^{1/2 + o(1)})$ parallel depth in total, thus we can simulate it over $G_T$ in $\O(D\alpha^{1/2 + o(1)}+\alpha)$ rounds. Thus by \cref{lemma:broadcast}, it is direct to check (by the results of \cref{sec:parallel}) that step 2 can be implemented with $\O(D\alpha^{1/2 + o(1)} + \alpha)$ rounds. 

As $G_T$ now has diameter $\alpha^{1/2 + o(1)}$ with high probability by the proof of correctness of algorithm \ParallelDiam, simulating a BFS and globally broadcasting the results takes $\O(D\alpha^{1/2+o(1)} + \alpha)$ rounds. Summing these contributions gives the desired bound.
\end{proof}
Now, we prove \cref{thm:congest}.
\begin{proof}[Proof of \cref{thm:congest}]
The proof simply involves substituting various values of $\alpha$ into \cref{thm:congestmain}. We handle the cases $D \le n^{1/4}$ and $D \ge n^{1/4}$ separately.
\paragraph{Case 1: $D \le n^{1/4}$.} We set $\alpha = \sqrt{n}.$ This gives us an algorithm that runs in $\O(\sqrt{n} + Dn^{1/4 + o(1)})$ rounds. It is easy to check that $Dn^{1/4 + o(1)} \le D^{2/3} n^{1/3 + o(1)}$ for $D \le n^{1/4}$, which completes the argument for this case.
\paragraph{Case 2: $D \ge n^{1/4}$.} In this case, we set $\alpha = \frac{n^{2/3}}{D^{2/3}}.$ It is easy to check that $D \ge n^{1/4}$ gives that $\alpha \le \sqrt{n}.$ Therefore, the resulting algorithm runs in $\O\left(n/\alpha  + D\alpha^{1/2 + o(1)} \right) = \O(D^{2/3}  n^{1/3+o(1)})$ rounds as desired.
\end{proof}

{\small
\bibliographystyle{alpha}
\bibliography{refs}}

\begin{appendix}
% \section{Omitted Proofs}

% \begin{lemma}[Bridges only help]
% Let $r, b, k$ be arbitrary nonnegative integers such that $r+b \ge k \ge 1$. We have that
% \[ \frac{\binom{r}{k}}{\binom{r+b}{k}} \cdot \left( \frac{2}{k+1}r + b \right) \le \frac{2}{k+1}(r+b). \]
% \end{lemma}

% \begin{proof}
% We assume $r \ge k$, or else the conclusion is trivial.
% By rearranging the conclusion, we can see that it suffices to show that for any nonnegative integers $r, b, k$ such that $r+b \ge k \ge 1$, we have that
% \[ \frac{\binom{r}{k}}{\binom{r+b}{k}} \le \frac{r+b}{r+\frac{k+1}{2}b}. \]
% We prove the statement by induction on $k$. The case $k = 1$ reduces to $\frac{r}{r+b} \le 1$, which is trivial.
% Now, note that for $k \ge 1$, by induction we have that
% \begin{align*}
% \frac{\binom{r}{k+1}}{\binom{r+b}{k+1}} &= \frac{\binom{r}{k}}{\binom{r+b}{k}} \cdot \frac{r-k}{r+b-k} \\
% &\le \frac{r+b}{r+\frac{k+1}{2}b} \cdot \frac{r-k}{r+b-k} = \frac{r+b}{r+\frac{k+1}{2}b} \cdot \left(1 - \frac{b}{r+b-k} \right) \\
% &\le \frac{r+b}{r+\frac{k+1}{2}b} \cdot \left(1 - \frac{b}{r+\frac{k+3}{2}b} \right) = \frac{r+b}{r+\frac{k+3}{2}b} \\
% &\le \frac{r+b}{r+\frac{k+2}{2}b}
% \end{align*}
% as desired. We have used that $r+b-k \le r+\frac{k+3}{2}b$, which is trivial for $k \ge 1.$
% \end{proof}
\end{appendix}

\end{document}